%% file: main.tex
\documentclass{article}

\usepackage{fullpage}
\usepackage{natbib}
\usepackage[utf8]{inputenc}
\usepackage[T1]{fontenc}
\usepackage{hyperref}
\usepackage{url}
\usepackage{booktabs}
\usepackage{amsfonts}
\usepackage{nicefrac}
\usepackage{microtype}
\usepackage{xcolor}
\usepackage{multirow}
\usepackage{marginnote}
\usepackage{algorithm}
\usepackage{algorithmicx}
\usepackage[noend]{algpseudocode}
\usepackage{amsmath}
\usepackage{amssymb}
\usepackage{mathtools}
\usepackage{amsthm}
\usepackage[capitalize,noabbrev]{cleveref}
\usepackage{subcaption}
\usepackage{makecell}
\usepackage{numprint}
\npdecimalsign{.}
\usepackage{enumitem}
\usepackage{diagbox}

\theoremstyle{plain}
\newtheorem{theorem}{Theorem}[section]

\newtheorem{lemma}[theorem]{Lemma}
\newtheorem{corollary}[theorem]{Corollary}

\theoremstyle{definition}
\newtheorem{definition}[theorem]{Definition}

\newtheorem{example}[theorem]{Example}
\theoremstyle{remark}

\input{math_commands}

\allowdisplaybreaks

\algnewcommand{\LeftComment}[1]{\Statex \(\triangleright\) #1}

\DeclareRobustCommand{\vol}{\mathrm{vol}}
\DeclareRobustCommand{\opt}{\mathrm{OPT}}
\DeclareRobustCommand{\cost}{\mathrm{cost}}

\title{Learning-Augmented Streaming Algorithms \\ for Correlation Clustering}

\author{
Yinhao Dong\thanks{School of Computer Science and Technology, University of Science and Technology of China, Hefei, Anhui Province, China. Email: {yhdong@mail.ustc.edu.cn}}
\and
Shan Jiang\thanks{School of Computer Science and Technology, University of Science and Technology of China, Hefei, Anhui Province, China. Email: {js\_@mail.ustc.edu.cn}}
\and
Shi Li\thanks{School of Computer Science, Nanjing University, Nanjing, Jiangsu Province, China. New Cornerstone Science Laboratory. Email: {shili@nju.edu.cn}}
\and
Pan Peng\thanks{Corresponding author. School of Computer Science and Technology, University of Science and Technology of China, Hefei, Anhui Province, China. Email: {ppeng@ustc.edu.cn}}
}

\date{}
\begin{document}

\maketitle

\begin{abstract}
  We study streaming algorithms for Correlation Clustering. Given a graph as an arbitrary-order stream of edges, with each edge labeled as positive or negative, the goal is to partition the vertices into disjoint clusters, such that the number of disagreements is minimized. In this paper, we give the first learning-augmented streaming algorithms for the problem on both complete and general graphs, improving the best-known space-approximation tradeoffs. Based on the works of Cambus et al. (SODA'24) and Ahn et al. (ICML'15), our algorithms use the predictions of pairwise distances between vertices provided by a predictor. For complete graphs, our algorithm achieves a better-than-$3$ approximation under good prediction quality, while using $\tilde{O}(n)$ total space. For general graphs, our algorithm achieves an $O(\log |E^-|)$ approximation under good prediction quality using $\tilde{O}(n)$ total space, improving the best-known non-learning algorithm in terms of space efficiency. Experimental results on synthetic and real-world datasets demonstrate the superiority of our proposed algorithms over their non-learning counterparts.
\end{abstract}

\newpage
\section{Introduction}
\label{sec:intro}
Correlation Clustering is a fundamental problem in machine learning and data mining, and it has a wide range of applications, such as image segmentation \citep{KYNK14}, community detection \cite{SDELM21,Veldt18community}, automated labeling \citep{CKP08}, etc. 
Given a graph $G=(V,E=E^+\cup E^-)$, where each edge is labeled as positive ($+$) or negative ($-$), the goal is to find a clustering $\mathcal{C}$, i.e., a partition of $V$ into disjoint clusters $C_1,C_2,\dots,C_t$, where $t$ is arbitrary, that minimizes the following cost:
\begin{align}
\label{eq:cost}
   \operatorname{cost}_G(\mathcal{C}) := |\{(u,v)\in E^+: \exists i\neq j,~u\in C_i, v\in C_j \}| 
   +|\{(u,v)\in E^-: \exists i,~u,v\in C_i \}|.
\end{align}
That is, the number of positive edges between different clusters, plus the number of negative edges within the same cluster. (We often refer to this as the number of \emph{disagreements}.)

The most commonly studied version of this problem is on complete graphs, introduced by Bansal et al.~\cite{BBC04}, and known to be NP-hard and even APX-hard~\citep{CGW05}. Hence, significant efforts have been dedicated to designing approximation algorithms for this setting \citep{BBC04,CGW05,ACN08,CMSY15,CLN22,CLLN23,CCLLNV24,CLPTYZ24,CCL+25}, culminating in a $1.437$-approximation via a linear program (LP) based rounding \citep{CCLLNV24,CCL+25}. 
In contrast, the problem on general graphs has received less attention. Charikar et al.~\cite{CGW05} and Demaine et al.~\cite{DEFI06} proposed an $O(\log n)$-approximation algorithm via ball-growing based LP rounding. They also showed that this version is equivalent to the minimum multicut problem, and is thus APX-hard and unlikely to admit better-than-$\Theta(\log n)$ approximations.

Partially due to storage limitations and the rapidly growing volume of data, \emph{streaming algorithms} for Correlation Clustering have received increasing attention recently.
In this setting, a graph is represented as a sequence of edge insertions or deletions, known as a \emph{graph stream}. The objective is to scan the sequence in a few number of passes, ideally, $1$ pass and find a high-quality clustering of the vertex set, 
while minimizing space usage. If the sequence contains only edge insertions, it is referred to as an \emph{insertion-only} stream; if both insertions and deletions are allowed, it is referred to as a \emph{dynamic} stream. Since the output of the clustering inherently requires $\Omega(n)$ space (as each vertex needs a label to indicate its cluster membership), most previous research has primarily focused on the \emph{semi-streaming} model~\cite{Feigenbaum05semi}, i.e., the algorithm is allowed to use $\tilde{O}(n):= O(n\operatorname{polylog}n)$ space.\footnote{On the other hand, Assadi et al.~\cite{ASW23} studied streaming algorithms using $\operatorname{polylog}n$ bits of space for estimating the optimum Correlation Clustering \emph{cost}, while their algorithms do \emph{not} find the clustering.}
We note that the space used by the streaming algorithm can be further divided into the space for updating data structures during the stream and the space for post-processing.

A long line of prior work has focused on the complete graph setting ~\cite{ACGMW21,CLMNP21,AW22,BCMT22,BCMT23,CM23,CKLPU24,CLPTYZ24,AKP25}. 
If the total space of the algorithm is restricted to $\tilde{O}(n)$, then
the best-known space-approximation tradeoff in dynamic streams is a $(3+\varepsilon)$-approximation algorithm~\cite{CKLPU24}, which uses $\tilde{O}(\varepsilon^{-1}n)$ total space. 
If only the space used for updating during the stream is restricted to $\tilde{O}(n)$, then
the best-known space-approximation tradeoff in dynamic streams is achieved by an $(\alpha_{\mathrm{BEST}}+\varepsilon)$-approximation algorithm~\cite{AKP25}, where $1.042 \le \alpha_{\mathrm{BEST}} \le 1.437$~\cite{CCLLNV24} denotes the best approximation ratio of any polynomial-time classical algorithm. However, this algorithm does not bound the space for post-processing, which can be significantly larger.
For general graphs, the only known dynamic streaming algorithm is due to Ahn et al.~\cite{ACGMW21}, which uses the multiplicative weight update method on the sparsified graph to achieve a $3(1+\varepsilon)\log|E^-|$-approximation while using $\tilde{O}(\varepsilon^{-2}n+|E^-|)$ space.

In this paper, we explore new approaches that enable improved space-approximation tradeoffs on both complete and general graphs by leveraging ideas from \emph{learning-augmented algorithms}. A learning-augmented algorithm uses \emph{predictions} to enhance its performance. These algorithms stem from practical scenarios where machine learning techniques exploit data structure to exceed the worst-case guarantees of traditional algorithms. 
Our learning-augmented algorithms fit into the category of learning-augmented streaming algorithms \citep{frequency,JLLRW20,CEILNRSWWZ22,ACNSV23,Dong25maxcut,ACGSW25}.
It is worth mentioning that both our work and previous efforts on learning-augmented streaming algorithms mainly focus on using predictors to improve the corresponding space-accuracy tradeoffs.

Now, we describe the prediction we are considering. We assume that the algorithm has oracle access 
to a predictor $\Pi: \tbinom{V}{2} \rightarrow [0,1]$ that predicts the \emph{pairwise distances}
\footnote{Note that one can directly treat $1-d_{uv}$ as the pairwise \emph{similarity} between $u$ and $v$.} 
$d_{uv}$ between any two vertices $u$ and $v$ in $V$. 
We believe such predictors are natural and arise in many situations. 

\begin{example}[Multiple graphs on the same vertex set]
It is common to define \emph{multiple} graphs over the same vertex set. For example, in healthcare, patients can be represented as vertices, and different networks may capture relationships such as shared medical conditions (disease networks), visits to the same providers (provider networks), or participation in clinical trials. In biology, vertices might represent genes or proteins, with different networks reflecting protein-protein interactions, gene co-expression, or signaling pathways.
Machine learning or data mining techniques can then be used to learn pairwise distances between nodes across these networks. If two patients or genes/proteins are similar in one network, they often exhibit similar behavior in others.
\end{example}

\begin{example}[Temporal graphs] 
A similar situation occurs in temporal graphs, where a sequence of graphs shares the same vertex set but has different edge sets over time. Pairwise distances learned from past graphs can help extract structural insights in the present or future.
\end{example}

\emph{Leveraging these distances across networks can greatly aid in exploring the cluster structure of any newly defined network over the same vertex set}. 
We also remark that several other works have considered similar oracles for pairwise distance in different contexts, such as the query model~\citep{Silwal23kwikbucks, kuroki2024query}.

\subsection{Our results}

Our results are summarized in \cref{tab:prior-work}. Using the above predictions, we give the first learning-augmented streaming algorithms for Correlation Clustering on both complete and general graphs. Specifically, for complete graphs, our algorithm beats the $3$-approximation if the predictions are good, while still achieving a $(3+\varepsilon)$-approximation even if the predictor behaves poorly. For general graphs, our algorithm achieves an $O(\log |E^-|)$-approximation under good quality while using less space than the existing non-learning algorithm.   Furthermore, our algorithms are simple and easily implementable. 
We will use a parameter $\beta \ge 1$ to measure the quality of our predictor. 
Informally, we call a predictor \emph{$\beta$-level} if the cost of the clustering induced by the predictions is at most a $\beta$ factor of the cost of the optimal solution. 
We refer to \cref{sec:predictor} for the formal definition. 
In short, the smaller $\beta$ is, the higher the quality of the predictor. 

\begin{table*}[t!]
\centering
\caption{Comparison of our main results with the best-known space-approximation tradeoffs in polynomial time. Here, $n=|V|,\varepsilon \in (0,1)$ and $\beta \geq 1$. All algorithms are single-pass, and space is measured in words.}
\label{tab:prior-work}
\renewcommand{\arraystretch}{1.1}
\begin{tabular}{|c|c|c|}
\hline
\textbf{Setting} & \makecell{\textbf{Best-known space-approximation} \\ \textbf{tradeoffs (without predictions)}}  & {\bf Our results}  \\ \hline
\multirow[c]{4}{*}{\makecell{complete graphs,\\ dynamic streams}}  & \makecell{$(3+\varepsilon)$-approximation,\\ $\tilde{O}(\varepsilon^{-1}n)$ total space  
\citep{CKLPU24}}  & \multirow[c]{4}{*}{\makecell{$(\min\{2.06\beta, 3\}+\varepsilon)$-approximation,\\ $\tilde{O}(\varepsilon^{-2}n)$ total space (Theorem~\ref{thm:main-result-dynamic})}} \\ \cline{2-2}
& \makecell{$(\alpha_{\mathrm{BEST}}+\varepsilon)$-approximation,\\ $\tilde{O}(\varepsilon^{-2}n)$ space for updating,\\$\operatorname{poly}(n)$ space for post-processing \citep{AKP25}} & \\ \hline
\makecell{general graphs,\\ dynamic streams}  & \makecell{$O(\log|E^-|)$-approximation,\\ $\tilde{O}(\varepsilon^{-2}n+|E^-|)$ total space  \citep{ACGMW21}} & \makecell{$O(\beta \log |E^-|)$-approximation,\\ $\tilde{O}(\varepsilon^{-2}n)$ total space (Theorem~\ref{thm:main-result-general})} \\ \hline
\end{tabular}
\end{table*}

For the complete graph setting, we have the following theorem. (In the following and throughout the paper, ``with high probability'' refers to the probability of at least $1-1/n^c$ for some constant $c>0$.)

\begin{theorem}
\label{thm:main-result-dynamic} 
Let $\varepsilon\in (0,1/4)$ and $\beta \geq 1$. Given oracle access to a $\beta$-level predictor, there exists a single-pass streaming algorithm that, with high probability, achieves an expected $(\min\{2.06\beta, 3\}+\varepsilon)$-approximation for Correlation Clustering on complete graphs in dynamic streams. The algorithm uses $\tilde{O}(\varepsilon^{-2}n)$ words of space. 
\end{theorem}

Note that our algorithm achieves a better-than-$3$ approximation in dynamic streams under good prediction quality, whereas the previous best-known semi-streaming algorithm in dynamic streams is a $(3+\varepsilon)$-approximation due to Cambus et al.~\cite{CKLPU24}.
That is, our algorithm is both consistent and robust, 
as desired for most natural learning-augmented algorithms \citep{MV21}.
We further highlight that the recent $(\alpha_{\mathrm{BEST}}+\varepsilon)$-approximation algorithm~\cite{AKP25} does not bound the space usage during the post-processing phase, which may be larger than $\tilde{O}(n)$.
On the other hand, while there exists a single-pass $(1+\varepsilon)$-approximation algorithm for dynamic streams \citep{ACGMW21,BCMT23}, it requires exponential post-processing time and is thus impractical.
In contrast, our algorithm uses polynomial post-processing time.

Furthermore, we also obtain an algorithm for complete graphs in insertion-only streams (see \cref{sec:omitted-details-insertion}), which differs from our algorithm in dynamic streams but achieves the same approximation ratio. It is simpler and more practical than the existing $1.847$-approximation algorithm~\cite{CLPTYZ24}. 

For general graphs, we present the following result on the approximation-space trade-off for streaming Correlation Clustering, using a slightly different type of predictor.

\begin{theorem}
\label{thm:main-result-general} 
Let $\varepsilon\in (0,1/4)$ and $\beta \geq 1$. Given oracle access to an adapted 
$\beta$-level predictor, there exists a single-pass streaming algorithm that, with high probability, achieves an $O(\beta \log |E^-|)$-approximation for Correlation Clustering on general graphs in dynamic streams. The algorithm uses $\tilde{O}(\varepsilon^{-2}n)$ words of space.
\end{theorem}

We remark that the above algorithm can be extended to work under the previous notion of predictors for a broad class of graphs (see \Cref{cor:general-reff}). Note that the best-known streaming algorithm for general graphs is an $O(\log|E^-|)$-approximation while using $\tilde{O}(\varepsilon^{-2}n+|E^-|)$ words of space~\cite{ACGMW21}. In contrast, our learning-augmented algorithm attains a comparable approximation ratio under good prediction quality, while using less space. 
We also note that it is standard to assume that the space required to implement/store the predictor is \emph{not} included in the space usage of our algorithms, as is common in learning-augmented streaming algorithms \citep{frequency,JLLRW20,CEILNRSWWZ22,ACNSV23,Dong25maxcut,ACGSW25}.

To complement our theoretical results, we conduct comprehensive experiments to evaluate our algorithms on both synthetic and real-world datasets. Experimental results demonstrate the superiority of our learning-augmented algorithms. All the missing proofs are deferred to appendix. 

\subsection{Technical overview}
\subsubsection{Technical overview of our algorithms for complete graphs}
Our algorithms for complete graphs (i.e., \cref{alg:dynamic-stream} and \cref{alg:insertion-only}) rely on the influential \textsc{Pivot} algorithm by Ailon et al.~\cite{ACN08} and the LP rounding algorithm by Chawla et al.~\cite{CMSY15}. The \textsc{Pivot} algorithm begins by selecting a random permutation $\pi$ over the vertices of the graph. It then iteratively forms clusters by choosing the vertex with the smallest rank according to $\pi$, along with its neighbors in the graph. Once a cluster is formed, it is removed from the graph. This process continues until all vertices have been assigned to clusters. The LP rounding algorithm first solves an LP corresponding to Correlation Clustering, and then applies a \textsc{Pivot}-based algorithm, using the LP solution to form clusters.

Next, we describe our algorithms. The high-level idea is to incorporate the above LP rounding approach with the ``truncated'' \textsc{Pivot} algorithms \citep{CKLPU24,CM23}, where our predictions correspond to a feasible LP solution in some sense. Specifically, for dynamic streams, we maintain a certain number of $\ell_0$-samplers during the stream and use them to derive a truncated subgraph at the end of the stream. Then we run the \textsc{Pivot} algorithm and the LP rounding algorithm on the subgraph respectively and obtain two clusterings. Finally, we output the clustering with the lower cost.
For insertion-only streams, we employ two different methods to store at most $k$ neighbors for each vertex during the stream, 
then run the \textsc{Pivot} algorithm on the two stored subgraphs. We obtain two clusterings and output the one with the lower cost.

The analysis is non-trivial, even in insertion-only streams.
We categorize all clusters into pivot clusters and singleton clusters, and analyze their costs respectively.
Our key observation is that the truncated version of the LP rounding algorithm is equivalent to the algorithm that first samples a subgraph $G'$ according to the predictions and then runs the ``truncated'' \textsc{Pivot} algorithms on $G'$.   Our main technical contribution is to prove that
1) the cost of pivot clusters produced by the truncated version of the LP rounding algorithm is at most $2.06\beta$ times the cost of optimal solution (\cref{lem:seq-pairwisediss-pivot} and \cref{lem:pivot-cost}); 
2) the optimal solution on $G'$ does not differ from the optimal solution on the original graph $G$ by a lot (\cref{lem:opt-opt'}). In this way, our algorithms can keep the space small while achieving an approximation ratio better than $3$ under good prediction quality.

\subsubsection{Technical overview of our algorithms for general graphs}
Our algorithm for general graphs (i.e., \cref{alg:general}) is inspired by the $O(\log |E^-|)$-approximation algorithm by Ahn et al.~\cite{ACGMW21}, which sparifies the positive subgraph $G^+$ to $H^+$, stores all negative edges $E^-$ during the stream, and applies the multiplicative weight update method to approximately solve a linear program. The resulting LP solution then guides an influential ball-growing procedure \cite{CGW05,DEFI06} on the sparsifier $H^+$.

The high-level idea of our algorithm is to incorporate the above algorithm with our pairwise distance predictions, which, in a sense, form a feasible LP solution. Specifically, we also maintain a sparsifier $H^+$ for $G^+$ during the stream and, in the post-processing phase,  perform ball-growing on $H^+$ using the predictions as distance metrics. Notably, we no longer need to store $E^-$ during the stream, leading to improved space complexity.

However, the approximation guarantee of this straightforward approach includes a $\log n$ term, whereas our goal is to replace it with a tighter $\log |E^-|$ term. This motivates us to further refine the algorithm. Specifically, during the stream, we maintain the sparsifier $H^+$ as before, while simultaneously storing the arriving negative edges and tracking their space usage. If at any point the space used to store negative edges exceeds $\tilde{O}(\varepsilon^{-2}n)$, we immediately stop storing them. After the stream ends (at which point the sparsifier is ready), we proceed with the ball-growing procedure as described above. Otherwise (i.e., if the space used by the negative edges never exceeds the threshold), we run the post-processing phase of the algorithm by Ahn et al.~\cite{ACGMW21}. In this way, our algorithm keeps the space small while achieving a near $O(\log |E^-|)$-approximation under good prediction quality.

\subsection{Further related work}

\textbf{Correlation Clustering.} 
In this paper, we study streaming algorithms for Correlation Clustering.
In the era of big data, other sublinear models for this problem have also received considerable attention in recent years, including sublinear-time algorithms~\cite{AW22,CLPTYZ24,CCL+25}, the Massively Parallel Computation (MPC) model~\cite{CLMNP21,BCMT22,CHS24,DMM24,CLPTYZ24,CCL+25}, and vertex/edge fully dynamic models~\cite{BCCGM24,DMM24}.
From another perspective, we focus on the minimization version of Correlation Clustering, which aims to minimize the number of disagreements. 
Other variants of this problem have also been studied, such as maximizing agreements~\cite{Swamy04}, or minimizing certain norms—or all norms—of the disagreement vector~\cite{PM18,CGS17,KMZ19,HIA24,DMN24,CLY25}.

\textbf{Learning-augmented algorithms.} Learning-augmented algorithms, also known as algorithms with predictions, have been actively studied in the context of online algorithms \citep{Purohit18improving,Bamas20primal,Scheduling,Im21online,caching,Antoniadis23online,Antoniadis23secretary,Angelopoulos24online}, data structures \citep{Mitzenmacher18model,FV20,Vaidya21partitioned,LLW22,Sato23fast}, graph algorithms \citep{Dinitz21faster,faster-graph,Banerjee23graph,LSV23,Davies23predictive,liu23predicted,Brand24dynamic,Henzinger24complexity,DePavia24learning}, sublinear-time algorithms \citep{Indyk19learning,support,Li23learning},  approximation algorithms \citep{ICLR22,ICLR23,CDGLP24,Ghoshal24constraint,Braverman24learning,AEPV25}, and mechanism design~\cite{XuL22,GkatzelisKST22,AgrawalBGOT22,LuW024,Caragiannis024,Colini-Baldeschi24}, among others.
Our work fits into the category of learning-augmented (graph)
streaming algorithms.

\section{Preliminaries}
\label{sec:preliminaries}
\textbf{Notations.} Throughout the paper, we let $G=(V,E)$ be an undirected and unweighted graph with $|V|=n$, $|E|=m$, where each edge is labeled as positive or negative (i.e., $E=E^+ \cup E^-$). 
In some places of the paper, we identify the input graph only with the set of positive edges, i.e., $G^+=(V,E^+)$ and the negative edges are defined implicitly.
For each vertex $u\in V$, let $N(u)$ be the set of all neighbors of $u$ and
$N^+(u)$ be the set of positive neighbors of $u$ (i.e., vertices that are connected by a positive edge).
Correspondingly, let $\deg(u):=|N(u)|$ be the degree of $u$, and similarly, $\deg^+(u):=|N^+(u)|$.
We use $\operatorname{cost}_G(\mathcal{C})$ to denote the cost of the clustering $\mathcal{C}$ on $G$, as defined in Eq.~(\ref{eq:cost}).
We say an algorithm achieves an $\alpha$-approximation if it outputs a clustering $\mathcal{C}$ on $G$ such that $\mathrm{OPT} \leq \operatorname{cost}_G(\mathcal{C}) \leq \alpha \cdot \mathrm{OPT}$, where $\mathrm{OPT}$ denotes the cost of an optimal solution on $G$.

We introduce additional technical preliminaries in \cref{sec:tools}.

\section{The $\beta$-level predictor}
\label{sec:predictor}
In this section, we give the formal definition of a $\beta$-level predictor.
\begin{definition}[$\beta$-level predictor]
\label{def:beta-level-predictor}
For any $\beta\geq 1$, we call a predictor \emph{$\beta$-level}, if it predicts the pairwise distances $d_{uv}\in [0,1]$ between any two vertices $u$ and $v$ in $G$ such that 
\begin{enumerate}[label=(\arabic*), ref=(\arabic*)]
    \item (triangle inequality) $d_{uv} + d_{vw} \ge d_{uw}$ for all $u, v, w\in V$,
    \item \label{item:beta-predictor} $\sum_{(u, v) \in E^{+}} d_{u v}+\sum_{(u, v) \in E^{-}}(1-d_{u v}) \leq \beta \cdot \mathrm{OPT}$. 
\end{enumerate}
\end{definition}

Intuitively, a smaller $\beta$ indicates a higher-quality predictor, and in this case $d_{uv}$ can be used to determine how likely $u$ and $v$ are in the same cluster of the optimal solution. However, we point out that the predictions can be completely independent of the input graph. In the worst case, the predictions can be arbitrarily bad, which is allowed for learning-augmented algorithms since robustness is a desired goal.
We remark that \cref{def:beta-level-predictor} is inspired by the metric LP formulation of Correlation Clustering. In a sense, the $\beta$-level predictor corresponds to a feasible solution to the LP.

Furthermore, for a general graph, we say a pairwise distance predictor is 
an \emph{adapted $\beta$-level predictor} if it satisfies the triangle inequality and $\sum_{(u,v)\in E_H^+}w'_{uv}d_{uv}+\sum_{(u,v)\in E^-}(1-d_{uv}) \le \beta \cdot \opt$, where $H^+:=(V,E_H^+,w')$ is an $\varepsilon$-spectral sparsifier of $G^+=(V,E^+)$, which approximates all the cuts in $G^+$ within a $(1\pm \varepsilon)$ factor.

\textbf{Practical consideration of the predictor.} As mentioned earlier, to cluster a graph, we may use ML models such as graph neural networks (GNNs) to learn pairwise distances from related networks defined on the same vertex set. These models can be trained to extract meaningful distances, for example, by learning node embeddings that map vertices to a Euclidean space. The distances between these embeddings naturally serve as pairwise distances and satisfy the triangle inequality. Although the second condition in \cref{def:beta-level-predictor} is a technical requirement for theoretical analysis, our algorithms remain applicable even when the given pairwise distances do not strictly satisfy this condition. In practice, as long as the distances are meaningful, they can be directly incorporated into our framework. We refer to \cref{sec:exp} for empirical results.

\section{Our algorithm for complete graphs in dynamic streams}
\label{sec:alg-dynamic}

\subsection{Offline version}
\label{subsec:non-streaming}
\textbf{Overview.}
To better illustrate the algorithmic ideas, we first describe the offline version of our algorithm (see \cref{alg:parallel-pivot} in \cref{sec:omitted-pseudocodes-dynamic}). The overall framework is similar to the algorithm proposed by Cambus et al.~\cite{CKLPU24}. Our algorithm takes $G^+=(V,E^+)$ as input.
Initially, we pick a random permutation $\pi$ over the vertices.
Then we divide all vertices into interesting and uninteresting vertices based on the relationship between the rank and the positive degree of a vertex. Specifically, a vertex $u$ is uninteresting if $\pi_u \geq \tau_u$ where $\tau_u := \frac{c}{\varepsilon}\cdot\frac{n\log n}{\deg^+(u)}$ (or equivalently $\deg^+(u)\geq \sigma_u$ where $\sigma_u := \frac{c}{\varepsilon}\cdot\frac{n\log n}{\pi_u}$), and interesting otherwise. Here, $\varepsilon \in (0,1/4)$ and $c$ is a universal large constant.
We run two pivot-based algorithms on the subgraph $G_\textnormal{store}$ induced by the set of interesting vertices and obtain two clusterings. Finally, we output the clustering with the lower cost.

Specifically, the two pivot-based algorithms used in the clustering phase are described as follows.

\textbf{Algorithm \textsc{TruncatedPivot}.}
This algorithm simulates the Parallel Truncated-Pivot algorithm by Cambus et al.~\cite{CKLPU24} and produces the same clustering.
This algorithm proceeds in iterations. Let $U^{(t)}$ denote the set of unclustered vertices in $G_\textnormal{store}$ at the beginning of iteration $t$.
Initially, all the interesting vertices are unclustered. At the beginning of iteration $t$, if $U^{(t)}\neq \emptyset$, then we pick the vertex $u$ from $U^{(t)}$ with the smallest rank. Then we mark it as a pivot and create a pivot cluster $S^{(t)}$ containing $u$ and all of its unclustered positive neighbors in $G_\textnormal{store}$. At the end of iteration $t$, we remove all vertices clustered in this iteration from $U^{(t)}$. Then the algorithm proceeds to the next iteration. If $U^{(t)}= \emptyset$ at the beginning of iteration $t$, then we know that all the interesting vertices are clustered. Now it suffices to assign each uninteresting vertex to a cluster. Each uninteresting vertex $u$ joins the cluster of pivot $v$ with the smallest rank if $(u,v)\in E^+$ and $\pi_v < \tau_u$. Then each unclustered vertex $u\in V$ creates a singleton cluster.
Finally, we output all pivot clusters and singleton clusters.
We defer its pseudocode (\cref{alg:pivot}) to \cref{sec:omitted-pseudocodes-dynamic}.

\textbf{Algorithm \textsc{TruncatedPivotWithPred}.}
This algorithm has oracle access to a $\beta$-level predictor $\Pi$. 
The algorithm closely resembles Algorithm \textsc{TruncatedPivot}. The differences are as follows: (1) At iteration $t$, we create a pivot cluster $S^{(t)}$ containing $u$ and add all the unclustered vertices $v$ in $G_\textnormal{store}$ to $S^{(t)}$ with probability $(1-p_{uv})$ independently, where $p_{uv}=f(d_{uv})$ and $d_{uv} = \Pi(u,v)$. If $(u,v)\in E^+$, then $f(d_{uv})=f^{+}(d_{uv})$; otherwise $f(d_{uv})=f^{-}(d_{uv})$.
We set $f^{+}(x)$ to be $0$ if $x<a$, $(\frac{x-a}{b-a})^2$ if $x \in[a, b]$, and $1$ if $x > b$, 
where $a=0.19$ and $b=0.5095$; we set $f^{-}(x)=x$. 
(2) Each uninteresting vertex $u$ joins the cluster of pivot $v$ in the order of $\pi$ with probability $(1-p_{uv})$ independently, if $\pi_v < \tau_u$. 
We defer its pseudocode (\cref{alg:CMSY-pivot}) to \cref{sec:omitted-pseudocodes-dynamic}. 

We have the following approximation guarantee of the offline algorithm.
\begin{lemma}
\label{lem:parallel} 
Let $\varepsilon\in (0,1/4)$ and $\beta \geq 1$. Given oracle access to a $\beta$-level predictor, the offline algorithm (\cref{alg:parallel-pivot}) achieves an expected $(\min\{2.06\beta, 3\}+\varepsilon)$-approximation.
\end{lemma}

\begin{algorithm}[t!]
\begin{algorithmic}[1]
\Require Graph $G^+ =(V,E^+)$ as an arbitrary-order dynamic stream of edges, oracle access to a $\beta$-level predictor $\Pi$
\Ensure Clustering/Partition of $V$ into disjoint sets
\LeftComment{\textbf{Pre-processing phase}}
\State Pick a random permutation of vertices $\pi: V \rightarrow \{1,\dots,n\}$.
\For{each vertex $u\in V$}
\State Let $\deg^+(u) \leftarrow 0$. Mark $u$ as unclustered and interesting.
\State Let $\sigma_u := \frac{c}{\varepsilon}\cdot \frac{n\log n}{\pi_u}$, where $c$ is a universal large constant.
\State Initialize $10c\log n\cdot \sigma_u$ independent $\ell_0$-samplers (with failure probability $1/10$) for the adjacency vector of $u$ (the row of the adjacency matrix of $G^+$ that corresponds to $u$).
\EndFor
\LeftComment{\textbf{Streaming phase}} 
\For{each item $(e_i=(u,v), \Delta_i\in \{-1,1\})$ in the dynamic stream} 
\State Update $\deg^+(u)$, $\deg^+(v)$ and all the $\ell_0$-samplers associated with $u$ and $v$.
\EndFor
\State Maintain an $\varepsilon$-spectral sparsifier $H^+$ for $G^+$ using the algorithm of \cref{thm:dynamic-spectral}.
\LeftComment{\textbf{Post-processing phase}}
\State A vertex $u$ marks itself uninteresting if $\deg^+(u)\geq \sigma_u$.
\State Retrieve all incident edges of interesting vertices (with high probability) using the $\ell_0$ samplers.
\State Let $G_\textnormal{store}$ be the graph induced by the interesting vertices.
\State $\mathcal{C}_1 \leftarrow \textsc{TruncatedPivot}(G^+, G_\textnormal{store},\pi)$ 
\State $\mathcal{C}_2 \leftarrow \textsc{TruncatedPivotWithPred}(G^+, G_\textnormal{store},\pi, \Pi)$
\State $\widetilde{\operatorname{cost}}_G(\mathcal{C}_1) \leftarrow \sum_{C \in \mathcal{C}_1} (\frac{1}{2}\partial_{H^+}(C)+\tbinom{|C|}{2} -\frac{1}{2} \sum_{u\in C}\deg^+(u) )$
\State $\widetilde{\operatorname{cost}}_G(\mathcal{C}_2) \leftarrow \sum_{C \in \mathcal{C}_2} (\frac{1}{2}\partial_{H^+}(C)+\tbinom{|C|}{2} -\frac{1}{2} \sum_{u\in C}\deg^+(u) )$
\State $i \leftarrow \arg\min_{i=1,2} \{\widetilde{\operatorname{cost}}_G(\mathcal{C}_i)\}$.
\State \Return $\mathcal{C}_i$
\end{algorithmic}
\caption{An algorithm for complete graphs in dynamic streams}
\label{alg:dynamic-stream}
\end{algorithm}

\subsection{Implementation in dynamic streams}
\label{subsec:imp-dynamic}
In this subsection, we implement the offline algorithm in dynamic streams, as shown in \cref{alg:dynamic-stream}.
A key observation is that it suffices to store the positive edges incident to interesting vertices since we apply pivot-based algorithms on the subgraph induced by interesting vertices and then try to assign uninteresting vertices to pivot clusters.
To this end, we maintain a certain number of $\ell_0$-samplers for each vertex, which can be achieved in dynamic streams \citep{JST11}. As we will see in the analysis, the $\ell_0$-samplers allow us to recover the edges incident to all the interesting vertices with high probability. Thus we can simulate the clustering phase of the offline algorithm. Specifically, we simulate \textsc{TruncatedPivot} and \textsc{TruncatedPivotWithPred} using the stored information, and finally output the clustering with the lower cost.

Note that in the final step, the cost of a clustering cannot be exactly calculated, as our streaming algorithm cannot store the entire graph. To overcome this challenge, we borrow the idea from \cite{BCMT23} and utilize the graph sparsification technique \cite{KMMMNST20fast} to estimate the clustering cost. 
Specifically, during the dynamic stream, we maintain an $\varepsilon$-spectral sparsifier $H^+$ for $G^+$ using the algorithm of \cref{thm:dynamic-spectral}. 
We also maintain the positive degree $\deg^+(u)$ for each vertex $u$. 
Then we can approximate the cost of a clustering up to a $(1\pm \varepsilon)$-multiplicative error with high probability. The formal proof of \cref{thm:main-result-dynamic} is deferred to \cref{sec:proof-main-dynamic}.

\subsection{Analysis of the offline algorithm}
\label{sec:analysis-parallel}

Since the final clustering returned by the offline algorithm is the one with the lower cost between those produced by the two pivot-based algorithms, we start by analyzing the costs of these two clusterings.
For ease of analysis, we separately examine the approximation ratios of the equivalent versions (Algorithms \textsc{CKLPU-Pivot} and \textsc{PairwiseDiss}) that produce these two clusterings.

\textbf{Algorithm \textsc{CKLPU-Pivot} \textnormal{(Algorithm~4 in \citep{CKLPU24})}.}
This algorithm proceeds in iterations. 
Let $U^{(t)}$ denote the set of unclustered vertices at the beginning of iteration $t$.
Initially, we pick a random permutation $\pi$ over vertices, and all the  vertices are unclustered. At the beginning of iteration $t$, let $\ell_t=\frac{c}{\varepsilon}\cdot \frac{n\log n}{t}$. Each unclustered vertex $v$ with $\deg^+(v) \geq \ell_t$ creates a \emph{singleton cluster}.
We pick the $t$-th vertex $u$ in $\pi$. If $u$ is unclustered, then we mark it as a pivot and create a \emph{pivot cluster} $S^{(t)}$ containing $u$ and all of its unclustered positive neighbors. At the end of iteration $t$, we remove all vertices clustered in this iteration from $U^{(t)}$ and proceed to the next iteration. 
Finally, we output all pivot clusters and singleton clusters.
We defer its pseudocode (\cref{alg:3-approx-dynamic}) to \cref{sec:omitted-pseudocodes-dynamic}.

\textbf{Algorithm \textsc{PairwiseDiss}.}
This algorithm has oracle access to a $\beta$-level predictor $\Pi$. The only difference from Algorithm \textsc{CKLPU-Pivot} is that at iteration $t$, we create a \emph{pivot cluster} $S^{(t)}$ containing $u$ and add all unclustered vertices $v$ to $S^{(t)}$ with probability $(1-p_{uv})$ independently, where $p_{uv}=f(d_{uv})$ and $d_{uv} = \Pi(u,v)$. 
We defer its pseudocode (\cref{alg:2.06-approx-dynamic}) to \cref{sec:omitted-pseudocodes-dynamic}.

\paragraph{The offline algorithm as a combination of \textsc{CKLPU-Pivot} and \textsc{PairwiseDiss}}
We first show that the offline algorithm can be equivalently viewed as a combination of Algorithms~\textsc{CKLPU-Pivot} and \textsc{PairwiseDiss}, assuming the same randomness is used.

\begin{lemma}[Lemma~8 in \cite{CKLPU24}]
\label{lem:eq-3-dynamic}
    If the offline algorithm (\cref{alg:parallel-pivot}) and \textsc{CKLPU-Pivot} use the same permutation $\pi$, then \textsc{TruncatedPivot} and \textsc{CKLPU-Pivot} output the same clustering. 
\end{lemma}

\begin{lemma}
\label{lem:eq-2.06-dynamic}
    If the offline algorithm (\cref{alg:parallel-pivot}) and \textsc{PairwiseDiss} use the same permutation $\pi$ and predictions $\{d_{uv}\}_{u,v\in V}$, then \textsc{TruncatedPivotWithPred} and \textsc{PairwiseDiss} output the same clustering with the same probability.
\end{lemma}

\subsubsection{The approximation ratios of \textsc{CKLPU-Pivot} and \textsc{PairwiseDiss}}
Now it suffices to analyze Algorithms~\textsc{CKLPU-Pivot} and \textsc{PairwiseDiss} separately. We follow the analysis framework in \cite{CKLPU24}. Specifically, we analyze the costs of pivot clusters and singleton clusters separately. For the former, we can directly apply the analysis of original pivot-based algorithms~\citep{ACN08,CMSY15}. Note that here we only need to focus on a subset of vertices (i.e., $V\setminus V_\textnormal{sin}$ where $V_\textnormal{sin}$ is the set of singletons). For the latter, we divide all the positive edges incident to singleton clusters (denoted as $E_\textnormal{sin}$) into good edges ($E_\textnormal{good}$) and bad edges ($E_\textnormal{bad}$). Specifically, we define a positive edge incident to a singleton cluster to be good if the other endpoint was included in a pivot cluster \emph{before} the singleton was created. Otherwise, the edge is bad. In other words, bad edges are those that either connect two singletons or the other endpoint was included in a pivot cluster \emph{after} the singleton was created.
The analysis in \citep{CKLPU24} shows that both the costs of good and bad edges can be charged to the pivot clusters, allowing us to bound the overall clustering cost.

The following lemma states the approximation guarantee of \textsc{CKLPU-Pivot}, and thus that of the clustering returned by \textsc{TruncatedPivot}.

\begin{lemma}[\cite{CKLPU24}]
\label{lem:dynamic-cost-3}
Let $\varepsilon\in (0,1/4)$.
     Let $\mathcal{C}_1$ denote the clustering  returned by \textsc{TruncatedPivot}, then $\E[\mathrm{cost}_G(\mathcal{C}_1)]\leq (3+12\varepsilon)\cdot \mathrm{OPT}+ \frac{1+4\varepsilon}{n^{\alpha-2}}$,
    where $\alpha := c/2-1 \gg 2$.
\end{lemma}

Next, we focus on the analysis of Algorithm~\textsc{PairwiseDiss}.
\begin{lemma}
\label{lem:seq-pairwisediss-pivot}
Let $P$ denote the cost of pivot clusters returned by Algorithm~\textsc{PairwiseDiss}. We have $\E[P]\leq  2.06\beta \cdot \mathrm{OPT}$.
\end{lemma}
\begin{proof}
Consider iteration $t$ of \textsc{PairwiseDiss}, if vertex $u$ considered in this iteration is unclustered (i.e., $u\in U^{(t)}$), then we call iteration $t$ a \emph{pivot iteration}.
The key observation is that the pivot iterations in \textsc{PairwiseDiss} are equivalent to the iterations of $2.06$-approximation LP rounding algorithm~\cite{CMSY15}: given that $u$ is unclustered (i.e., $u\in U^{(t)}$), the conditional distribution of $u$ is uniformly distributed in $U^{(t)}$, and the cluster created during this iteration contains $u$ and all the unclustered vertices $v$ added with probability $(1-p_{uv})$.
Therefore, we can directly apply the triangle-based analysis in~\cite{CMSY15}. 
Define $L:=\sum_{(u, v) \in E^{+}} d_{u v}+\sum_{(u, v) \in E^{-}}(1-d_{u v})$. Since the predictor is $\beta$-level, by \cref{def:beta-level-predictor}, the predictions $\{d_{uv}\}_{u,v\in V}$ satisfy triangle inequality and $L \leq \beta \cdot \mathrm{OPT}$.
It follows that for all pivot iterations $t$, $\E [P^{(t)}]\leq 2.06 \cdot\E[L^{(t)}]$, where $P^{(t)}$ is the cost induced by the pivot cluster created at iteration $t$, and $L^{(t)} := \sum_{(u, v) \in E^{+}\cap E^{(t)}} d_{u v}+\sum_{(u, v) \in E^{-}\cap E^{(t)}}(1-d_{u v})$ where $E^{(t)}$ is the set of edges \emph{decided} at iteration $t$, i.e., $E^{(t)}=\{(u,v)\in E: u,v\in U^{(t)}; u\in S^{(t)} \text{ or } v\in S^{(t)}\}$.
By linearity of expectation, we have
$\E[P]
    =\sum_{\text{$t$ is a pivot iteration}} \E[P_2^{(t)}]
    \leq 2.06\cdot L 
    \leq 2.06\beta \cdot \mathrm{OPT}.$ 
\end{proof}

\begin{corollary}
\label{cor:dynamic-cost-2.06}
     Let $\varepsilon\in (0,1/4)$. Let $\mathcal{C}_2$ denote the clustering  returned by \textsc{TruncatedPivotWithPred}. We have $\E[\mathrm{cost}_G(\mathcal{C}_2)]\leq (2.06\beta+8.24\beta\varepsilon)\cdot \mathrm{OPT}+ \frac{1+4\varepsilon}{n^{\alpha-2}}$,
    where $\alpha := c/2-1 \gg 2$.
\end{corollary}

We defer the proofs of 
\cref{lem:eq-2.06-dynamic}, 
\cref{cor:dynamic-cost-2.06} and, finally, \cref{lem:parallel} to \cref{sec:omitted-details-dynamic}.

\section{Our algorithm for general graphs in dynamic streams}
\label{sec:general}
\paragraph{Overview of the algorithm} 
Our algorithm is given in \cref{alg:general}.
The core of our algorithm builds upon the ball-growing framework in the work of Charikara et al.~\cite{CGW05} and Demaine et al.~\cite{DEFI06}. In our algorithm, we apply this framework to a sparsified graph and use the predictions $d_{uv}$ as distance metrics. Specifically, during the streaming phase, we maintain an $\varepsilon$-spectral sparsifier $H^+ := (V,E^+_H,w')$ for $G^+ = (V,E^+)$. 

At the same time, we store all arriving negative edges and track their space usage. If, at any point, this space exceeds $\tilde{O}(\varepsilon^{-2}n)$ words, we immediately stop storing negative edges. Once the stream ends (i.e., after $H^+$ has been constructed), we proceed to the post-processing phase and run the ball-growing procedure on $H^+$ to obtain the final clustering $\mathcal{C}_2$.
On the other hand, if the total space for storing negative edges remains within $\tilde{O}(\varepsilon^{-2}n)$ throughout the stream, we instead invoke the post-processing phase of the algorithm by Ahn et al.~\cite{ACGMW21} and return the resulting clustering $\mathcal{C}_1$.

Next we describe the ball-growing procedure, which proceeds iteratively. We initialize the clustering $\mathcal{C}_2=\emptyset$.
Let $R:= (V_R,E_R)$ denote 
the current graph at each iteration, where $V_R$ and $E_R$ are the sets of remaining vertices and edges, respectively. Initially, we set $V_R=V$ and $E_R = E_H^+$. At each iteration, we select an arbitrary vertex $u \in V_R$ as the center and initialize its radius as $r_u=0$. We gradually  increase $r_u$ to grow a ball around $u$ until a certain condition is met. Then we set all the vertices in the ball as a new cluster, remove them along with their incident edges from the current graph, and proceed to the next iteration. This process is repeated until no vertices are left.

More precisely, for a vertex $u \in V_R$ and a radius $r_u \ge 0$, we define the \emph{ball} centered at $u$ with radius $r_u$ as
    $B_d(u,r_u):=\{v\in V_R: d_{uv}\le r_u\}$,
which consists of all vertices in $R^+$ within distance at most $r_u$ from $u$, where distances are measured according to the predictions.

We next define the cut value and volume associated with a ball. Let $\partial_{H^+}(B_d(u,r_u))$ denote the value of the cut $(B_d(u,r_u),V\setminus B_d(u,r_u))$ in $H^+$:
\begin{align*}
    \partial_{H^+}(B_d(u,r_u)) := \sum_{\substack{(v,w)\in E_H^+:\\v\in B_d(u,r_u),w\notin  B_d(u,r_u)}}w'_{vw}.
\end{align*}
Furthermore, we define the \emph{volume} of the ball $B_d(u,r_u)$ in $H^+$, denoted by $\vol_{H^+}(B_d(u,r_u))$:
\begin{align*}
    \vol_{H^+}(B_d(u,r_u)):=\frac{V^*}{n} + \sum_{\substack{(v,w)\in E_H^+:\\v,w\in B_d(u,r_u)}}w'_{vw}d_{vw} + \sum_{\substack{(v,w)\in E_H^+:\\v\in B_d(u,r_u),w\notin  B_d(u,r_u)}}w'_{vw}(r_u-d_{uv}),
\end{align*}
where $V^*:= \sum_{(u,v)\in E_H^+}w'_{uv}d_{uv}$ denotes the total volume of the graph $H^+$.

The ball is finalized once the cut value induced by the ball is at most $O(\log n)$ times its volume. Specifically, we require that $\partial_{H^+}(B_d(u,r_u)) \le 3\ln(n+1)\cdot \vol_{H^+}(B_d(u,r_u))$. This condition is guaranteed by the following lemma:
\begin{lemma}[\cite{GVY96,CGW05,DEFI06}]
    For any vertex $u$, 
    there exists a radius $r_u < 1/3$ (which can be found in polynomial time) such that the corresponding ball $B_d(u,r_u)$ satisfies $\partial_{H^+}(B_d(u,r_u)) \le 3\ln(n+1)\cdot \vol_{H^+}(B_d(u,r_u))$.
\end{lemma}

\begin{algorithm}[t!]
\begin{algorithmic}[1]
\Require Graph $G =(V,E)$ as an arbitrary-order dynamic stream of edges, oracle access to an adapted $\beta$-level predictor $\Pi$
\Ensure Clustering/Partition of $V$ into disjoint sets
\LeftComment{\textbf{Streaming phase}} 
\State Maintain an $\varepsilon$-spectral sparsifier $H^+ := (V,E^+_H,w')$ for $G^+$ using the algorithm of \cref{thm:dynamic-spectral}.
\State Meanwhile, store all arriving negative edges and track their space usage. If the space ever exceeds $\tilde{O}(\varepsilon^{-2}n)$ words, then stop storing negative edges and, after the stream ends, \textbf{goto} \cref{line:ours}. \phantomsection \label{line:check-neg}
\LeftComment{\textbf{Post-processing phase}}
\State Run the post-processing phase of the algorithm by Ahn et al.~\cite{ACGMW21} and \Return the clustering $\mathcal{C}_1$. 
\State Let $V_R \leftarrow V$ and $E_R \leftarrow E^+_H$ denote the sets of remaining vertices and edges, respectively. \phantomsection \label{line:ours}
\State Let $\mathcal{C}_2 \leftarrow \emptyset$.
\State For any $u,v\in V$, $d_{uv} = \Pi(u,v)$.
\While{$V_R \neq \emptyset$}
\State Let $R:= (V_R,E_R)$ denote the current graph.
 Pick an arbitrary vertex $u \in V_R$. Let $r_u \leftarrow 0$.
\State Increase $r_u$ and grow a ball $B_d(u,r_u)$ on $R$ such that
    $\partial_{H^+}(B_d(u,r_u)) \le 3\ln(n+1)\cdot \vol_{H^+}(B_d(u,r_u))$.
\State $\mathcal{C}_2 \leftarrow \mathcal{C}_2 \cup B_d(u,r_u)$.
\State Remove the vertices in $B_d(u,r_u)$ from $V_R$ and the incident edges from $E_R$.
\EndWhile
\State \Return $\mathcal{C}_2$
\end{algorithmic}
\caption{An algorithm for general graphs in dynamic streams}
\label{alg:general}
\end{algorithm}

\subsection{Analysis of \cref{alg:general}}

\textbf{Space complexity.} 
The space complexity of \cref{alg:general} is dominated by maintaining the $\varepsilon$-spectral sparsifier and storing the negative edges during the streaming phase. By \cref{thm:dynamic-spectral} and the condition in \cref{line:check-neg}, the algorithm uses $\tilde{O}(\varepsilon^{-2}n)$ words of space.

\textbf{Approximation guarantee.}
Recall the condition in \cref{line:check-neg}: if the space used to store negative edges remains within $\tilde{O}(\varepsilon^{-2}n)$ throughout the stream, then the resulting clustering $\mathcal{C}_1$ is given by the algorithm of Ahn et al.~\cite{ACGMW21}, which has the following guarantee.
\begin{lemma}[\cite{ACGMW21}]
\label{lem:icml15}
$\cost_G(\mathcal{C}_1)\le 3(1+\varepsilon)\log |E^-| \cdot \opt =O(\log |E^-|)\cdot \opt$.
\end{lemma}

Otherwise, we have $|E^-| \ge n$, and the clustering $\mathcal{C}_2$ is obtained via the ball-growing procedure. We now analyze the cost of $\mathcal{C}_2$, which consists of two parts: the number of positive edges that cross between different clusters, and the number of negative edges that lie in the same cluster. In the following, we bound these two quantities separately.

Consider the weighted graph $H:= (V,E_H^+\cup E^-,\bar{w})$, where the edge weights are defined as $\bar{w}_{uv}=w'_{uv}$ for $(u,v)\in E_H^+$ and $\bar{w}_{uv}=1$ for $(u,v)\in E^-$.
The positive and negative costs of $\mathcal{C}_2$ on $H$, denoted $\cost_H^+(\mathcal{C}_2)$ and $\cost_H^-(\mathcal{C}_2)$ respectively, can be bounded by the following two lemmas.
\begin{lemma}
\label{lem:general-pos-cost}
$\cost_H^+(\mathcal{C}_2) \le 3\ln(n+1)\cdot \sum_{(u,v)\in E_H^+}w'_{uv}d_{uv}$. 
\end{lemma}

\begin{lemma}
\label{lem:general-neg-cost}
$\cost_H^-(\mathcal{C}_2)\le 3\sum_{(u,v)\in E^-}(1-d_{uv})$.
\end{lemma}

We defer the proofs of \cref{lem:general-pos-cost} and \cref{lem:general-neg-cost} to \cref{sec:omitted-general}.
Now we are ready to bound the cost of $\mathcal{C}_2$ on $G$.
\begin{lemma}
\label{lem:general-cost}
$\cost_G(\mathcal{C}_2)=O(\beta\log |E^-|)\cdot \opt$.
\end{lemma}
\begin{proof}
    To analyze the cost of $\mathcal{C}_2$ on the original graph $G$, we first show that the positive cost of $\mathcal{C}_2$ on $G$ is close to that on $H$.
    Since $H^+$ is an $\varepsilon$-spectral sparsifier of $G^+$, we have
        $\cost_G^+(\mathcal{C}_2)= \frac{1}{2}\sum_{C\in \mathcal{C}_2}\partial_{G^+}(C)
        \le \frac{1}{2}\sum_{C\in \mathcal{C}_2}\frac{1}{1-\varepsilon}\partial_{H^+}(C)
        =\frac{1}{1-\varepsilon} \cost_H^+(\mathcal{C}_2)$.
Therefore, by \cref{lem:general-pos-cost} and \cref{lem:general-neg-cost}, the cost of $\mathcal{C}_2$ on $G$ is
\begin{align*}
&\cost_G(\mathcal{C}_2)=\cost_G^+(\mathcal{C}_2)+\cost_G^-(\mathcal{C}_2)
= \cost_G^+(\mathcal{C}_2)+\cost_H^-(\mathcal{C}_2) 
\le \frac{1}{1-\varepsilon}\cost_H^+(\mathcal{C}_2)+\cost_H^-(\mathcal{C}_2) \\
\le~&\frac{3\ln(n+1)}{1-\varepsilon}\cdot \sum_{(u,v)\in E_H^+}w'_{uv}d_{uv} + 3\sum_{(u,v)\in E^-}(1-d_{uv})\\
\le~&(3+4\varepsilon)\ln(n+1)\cdot \sum_{(u,v)\in E_H^+}w'_{uv}d_{uv} + 3\sum_{(u,v)\in E^-}(1-d_{uv})\\
\le~&(3+4\varepsilon)\ln(n+1)\cdot \left(\sum_{(u,v)\in E_H^+}w'_{uv}d_{uv}+\sum_{(u,v)\in E^-}(1-d_{uv})\right)
=O(\beta \log |E^-|) \cdot \opt,
\end{align*}
where the second step follows from $\cost_G^-(\mathcal{C}_2) = \cost_H^-(\mathcal{C}_2)$, the third-to-last step uses the inequality $\frac{3}{1-\varepsilon} < 3+4\varepsilon$ for any $\varepsilon \in (0, 1/4)$, and the last step follows from the definition of the adapted $\beta$-level predictor and the fact that $|E^-| \ge n$.
\end{proof}

\begin{proof}[Proof of \cref{thm:main-result-general}]
    \cref{thm:main-result-general} follows from \cref{lem:icml15} and \cref{lem:general-cost}.
\end{proof}

Furthermore, if the input graph satisfies certain mild conditions, then the adapted $\beta$-level predictor can be relaxed to the standard $\beta$-level predictor (\cref{def:beta-level-predictor}).
We defer the details to \cref{subsec:d-regular}.

\section{Experiments}
\label{sec:exp}
In this section, we evaluate our proposed algorithm for complete graphs empirically on synthetic and real-world datasets. All  experiments are conducted on a CPU with an i7-13700H processor and 32 GB RAM.  
For all results, unless otherwise stated, we report the average clustering cost over $20$ independent trials.

\textbf{Datasets.}
\textbf{1) Synthetic datasets.}
These datasets are generated from the Stochastic Block Model (SBM).
We use this model to plant ground-truth clusters.
It samples positive edges between vertex pairs within the same cluster with probability $p>0.5$, and samples positive edges across different clusters with probability $(1-p)$. 
\textbf{2) Real-world datasets.} 
We use \textsc{EmailCore}~\citep{LKF07, YBLG17}, \textsc{Facebook}~\citep{ML12}, \textsc{LastFM}~\citep{Rozemberczki20characteristic}, and \textsc{DBLP}~\citep{Yang15defining} datasets.
 For simplicity, for all datasets, we only simulate insertion-only streams of edges. 
 We refer to \cref{subsec:detailed-dataset} for detailed descriptions of the datasets.

\textbf{Predictor descriptions.}
\textbf{1) Noisy predictor.} We use this predictor for datasets with available optimal clusterings. We form this predictor by performing perturbations on optimal clusterings. 
\textbf{2) Spectral embedding.} 
We use this predictor for \textsc{EmailCore} and \textsc{LastFM}. It first maps all vertices to $\mathbb{R}^d$ using the graph Laplacian, then clusters all vertices based on their embeddings. 
For any two vertices $u,v\in V$, we form the prediction $d_{uv}$ based on the spectral embeddings of $u,v$. 
\textbf{3) Binary classifier.} 
We use this predictor for datasets with available ground-truth communities. 
This predictor is constructed by training a binary classifier to predict whether two vertices belong to the same cluster using node features. 
The predictions (i.e., binary values in $\{0, 1\}$) are then used as pairwise distances $d_{uv}$ in our algorithms. 
We refer to \cref{subsec:detailed-predictor} for detailed descriptions of the predictors.

\textbf{Baselines.}
\textbf{1) $(3+\varepsilon)$-approximation non-learning counterparts.} 
For our algorithm for complete graphs in dynamic streams, the counterpart is Algorithm~\textsc{CKLPU24}~\citep{CKLPU24}.
\textbf{2) The agreement decomposition algorithm \textnormal{\textsc{CLMNPT21}~\citep{CLMNP21}}.} Though the approximation ratio in theory is large ($\approx 701$), this algorithm has been shown to give high-quality solutions in practice. Note that this algorithm  
requires multiple passes.
For fairness, all baselines were implemented with equal effort.

\begin{figure*}[t!]
	\centering
	\begin{subfigure}{0.244\textwidth}
		\includegraphics[width=\textwidth]{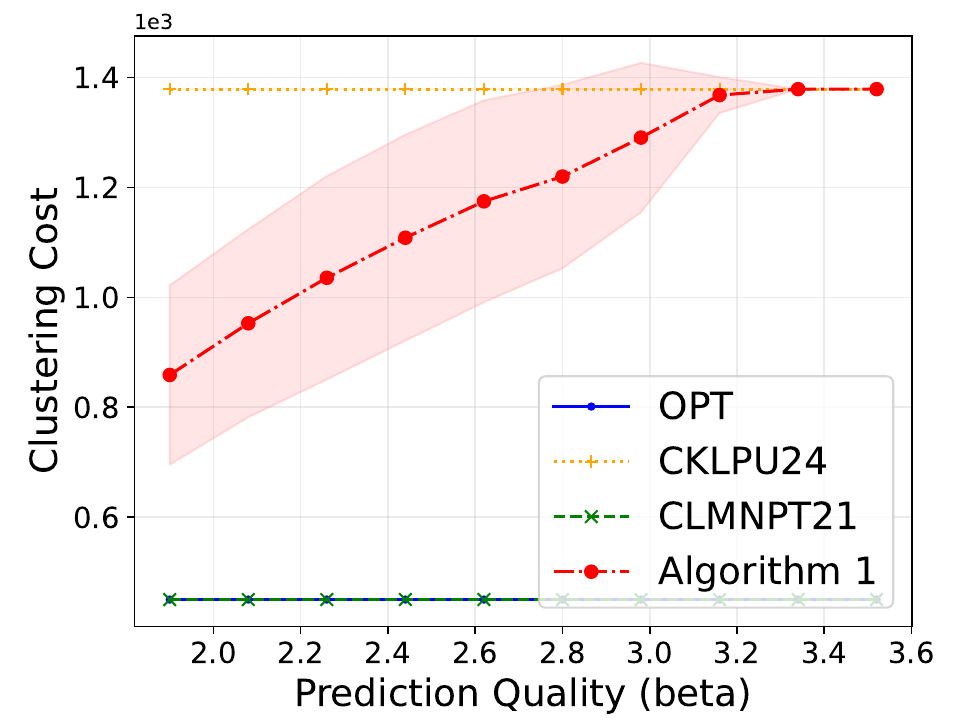}
        \captionsetup{justification=centering}
		\caption{$p=0.9$, vary $\beta$}
		\label{fig:synthetic-p-0.9}
	\end{subfigure}
	\hfill
	\begin{subfigure}{0.244\textwidth}
		\includegraphics[width=\textwidth]{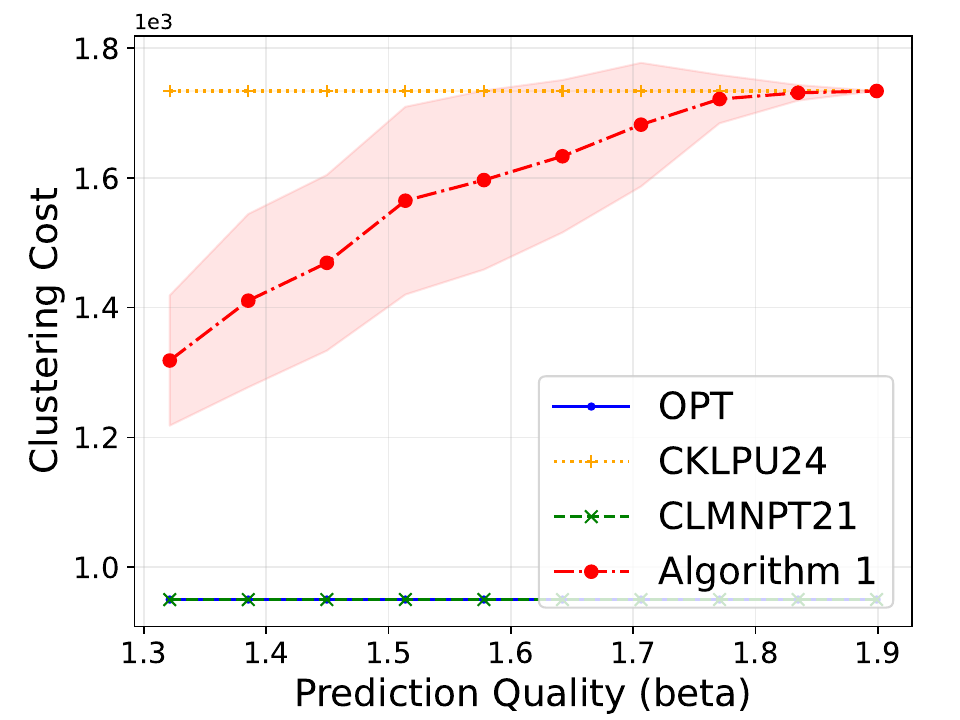}
        \captionsetup{justification=centering}
		\caption{$p=0.8$, vary $\beta$}
		\label{fig:synthetic-p-0.8}
	\end{subfigure}
        \hfill
	\begin{subfigure}{0.244\textwidth}
		\includegraphics[width=\textwidth]{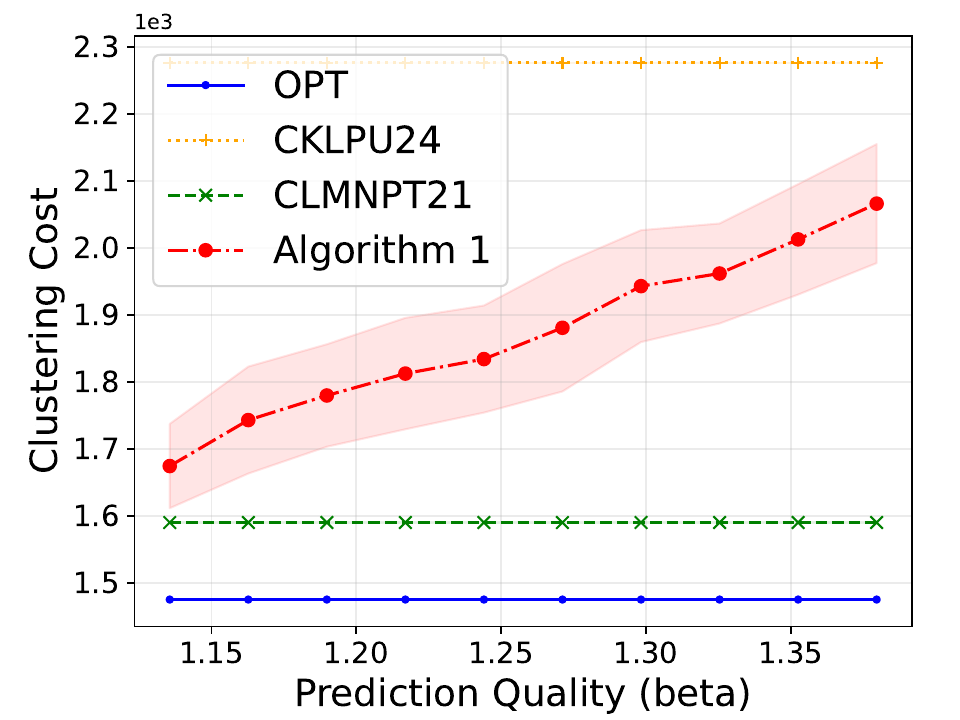}
        \captionsetup{justification=centering}
		\caption{$p=0.7$, vary $\beta$}
		\label{fig:synthetic-p-0.7}
	\end{subfigure}
	\hfill
        \begin{subfigure}{0.244\textwidth}
		\includegraphics[width=\textwidth]{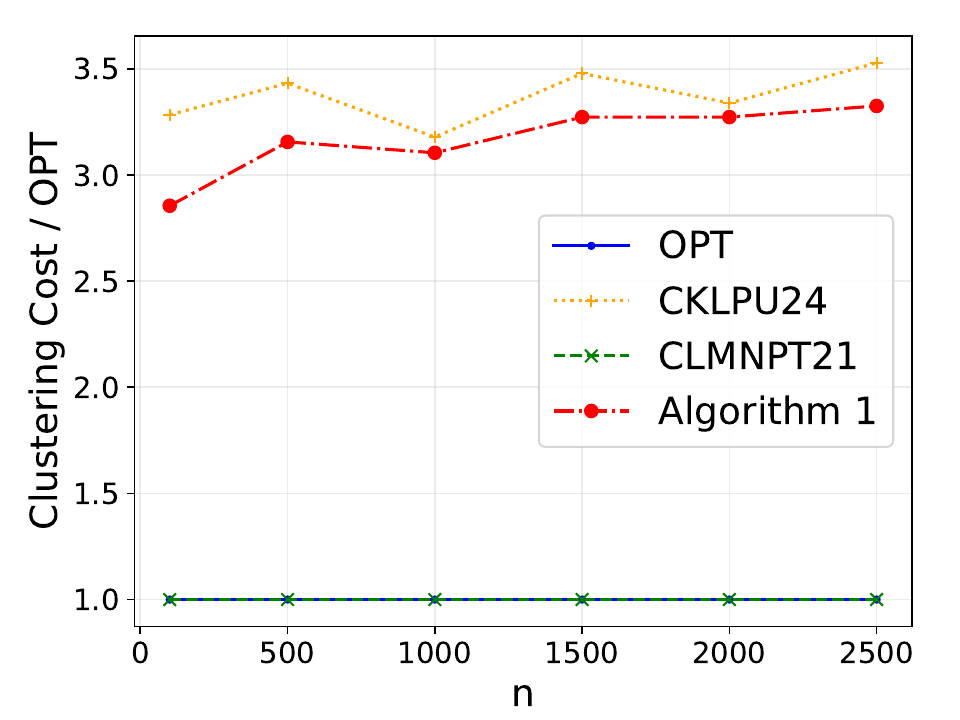}
		\captionsetup{justification=centering}
		\caption{$\beta \approx 2.8$, vary $n$}
		\label{fig:synthetic-vary-n-dynamic}
	\end{subfigure}
 
	\caption{Performance of \cref{alg:dynamic-stream} on synthetic datasets. We examine the effects of prediction quality $\beta$, SBM parameter $p$, and graph size $n$. We set $n=100$ in (\subref{fig:synthetic-p-0.9})--(\subref{fig:synthetic-p-0.7}) and $p=0.95$ in (\subref{fig:synthetic-vary-n-dynamic}).}
	\label{fig:synthetic}
\end{figure*}

\textbf{Results on synthetic datasets.}
\cref{fig:synthetic} shows the performance of \cref{alg:dynamic-stream} on synthetic datasets.  
\textbf{1) Varying $\beta$ and $p$.} We first examine the effect of $\beta$ and $p$ (see Figures~\ref{fig:synthetic}(\subref{fig:synthetic-p-0.9})--(\subref{fig:synthetic-p-0.7})). 
When $\beta$ is small, the cost of our algorithm is significantly lower than that of the $(3+\varepsilon)$-approximation non-learning baseline. Even for large $\beta$, our algorithm performs no worse.
Notably, we observe that the algorithm of \textsc{CLMNPT21} outputs (near-)optimal solution. We attribute this to the fact that the SBM graphs contain many dense components,  making them well-suited for this algorithm. Moreover, even when the ground-truth communities become less obvious (e.g., $p=0.7$), the clustering cost of \cref{alg:dynamic-stream} is reduced by up to $26\%$ compared to the algorithm of \textsc{CKLPU24}.
\textbf{2) Varying $n$.}
Furthermore, we investigate whether our algorithm scales well with graph size (see Figure~\ref{fig:synthetic}(\subref{fig:synthetic-vary-n-dynamic})). To clearly present our results, we calculate the ratio between the cost of each algorithm and the optimal solution. The result demonstrates that our algorithm performs well consistently as the graph size increases.

\begin{figure*}[t!]
	\centering
	\begin{subfigure}{0.244\textwidth}
		\includegraphics[width=\textwidth]{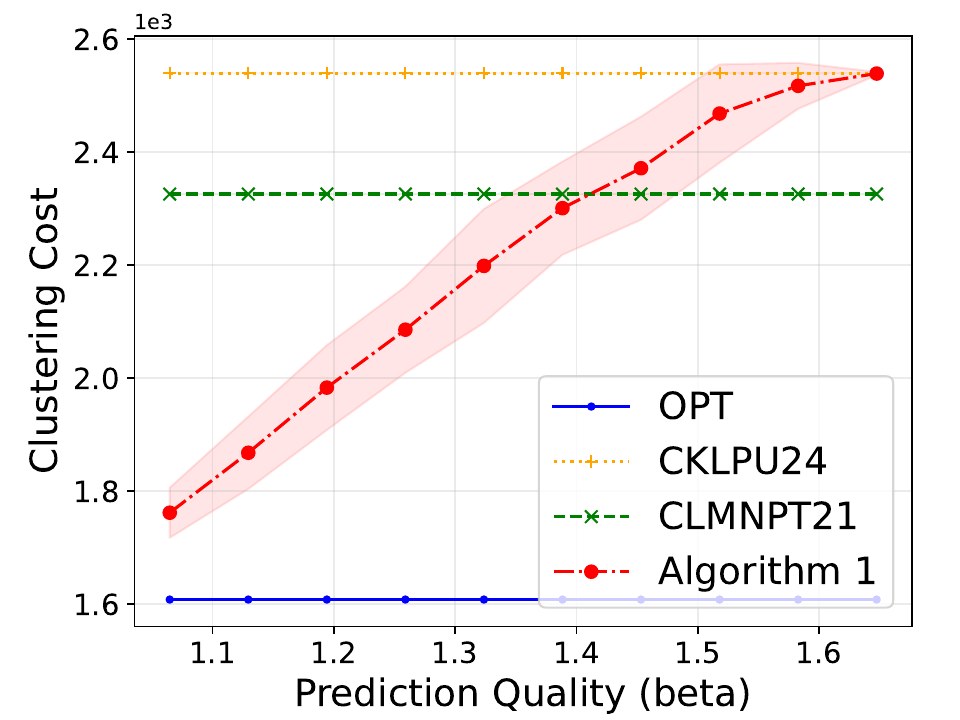}
		\caption{\textsc{FB 0}, vary $\beta$}
		\label{fig:fb0-vary-beta}
	\end{subfigure}
	\hfill
	\begin{subfigure}{0.244\textwidth}
		\includegraphics[width=\textwidth]{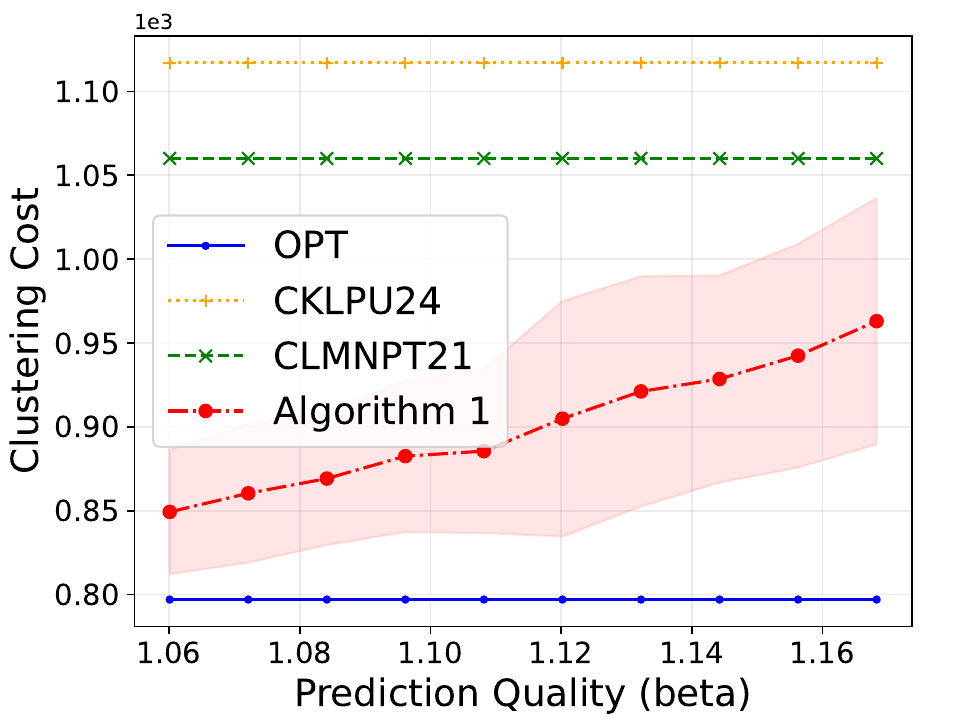}
		\caption{\textsc{FB 414}, vary $\beta$}
		\label{fig:fb414-vary-beta}
	\end{subfigure}
	\hfill
	\begin{subfigure}{0.244\textwidth}
		\includegraphics[width=\textwidth]{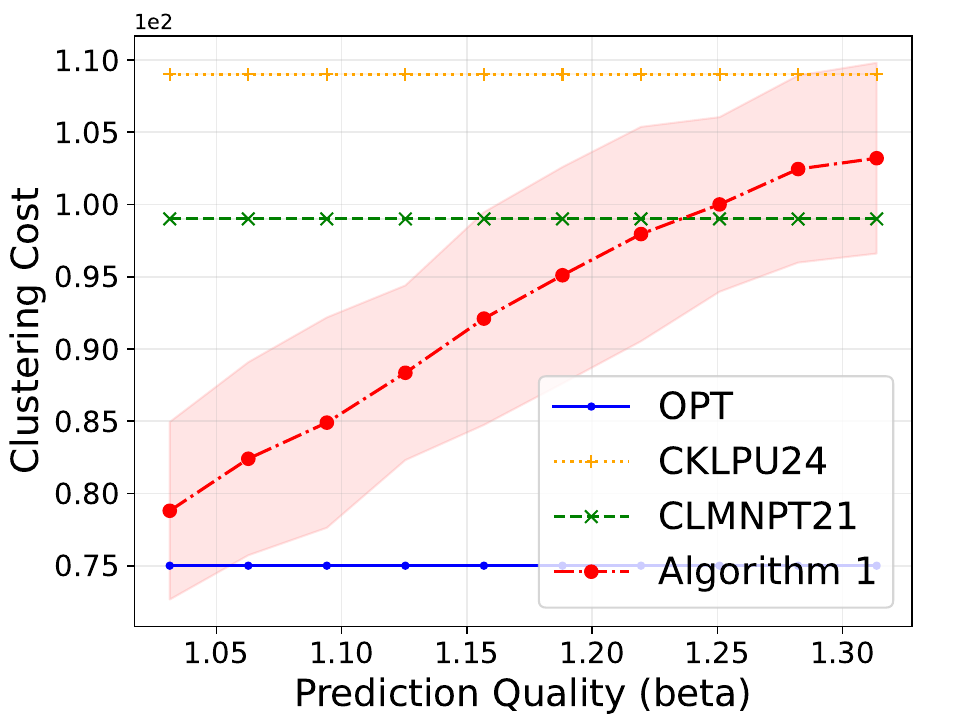}
		\caption{\textsc{FB 3980}, vary $\beta$}
		\label{fig:fb3980-vary-beta}
	\end{subfigure}
	\hfill
	\begin{subfigure}{0.244\textwidth}
		\includegraphics[width=\textwidth]{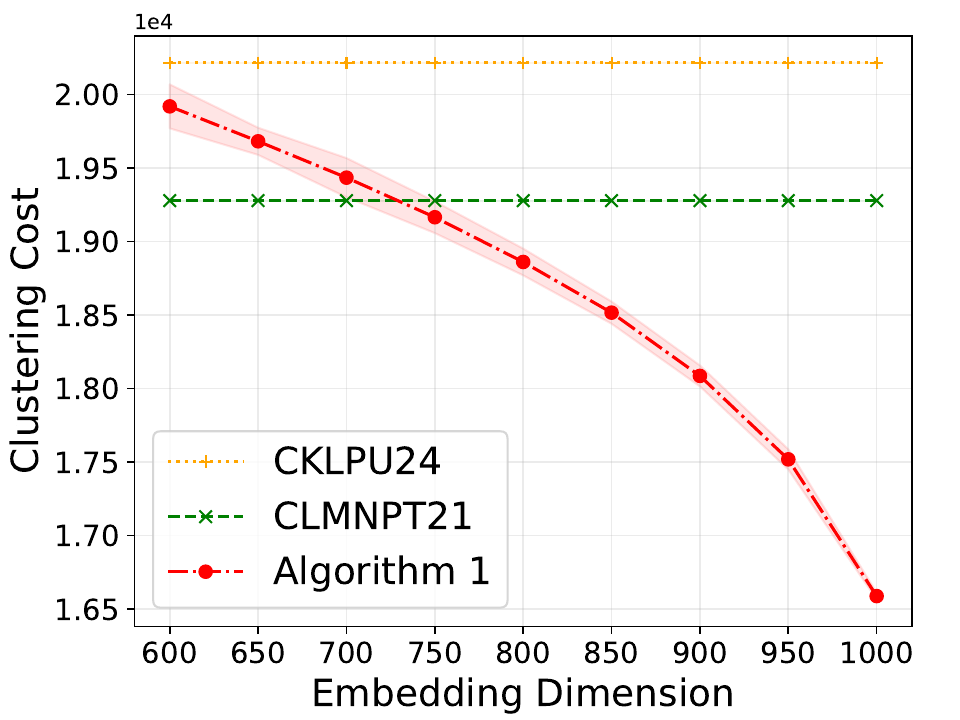}
		\caption{\textsc{EmailCore}, vary $d$}
		\label{fig:email-vary-d}
	\end{subfigure}
	\caption{Performance of \cref{alg:dynamic-stream} on real-world datasets. (\subref{fig:fb0-vary-beta})--(\subref{fig:fb3980-vary-beta}) show the effect of $\beta$ on three \textsc{Facebook} subgraphs. (\subref{fig:email-vary-d}) shows the effect of the dimension $d$ of spectral embeddings on \textsc{EmailCore}. Note that a larger $d$ indicates higher prediction quality (i.e., a smaller $\beta$).}
	\label{fig:results-realworld}
\end{figure*}

\textbf{Results on real-world datasets.}
\cref{fig:results-realworld} shows the performance of \cref{alg:dynamic-stream} on real-world datasets. 
The results demonstrate that under good prediction quality, \cref{alg:dynamic-stream} consistently outperforms other baselines across all datasets used. For example, in Figure~\ref{fig:results-realworld}(\subref{fig:fb0-vary-beta}), when $\beta \approx 1.2$, 
the average cost of our algorithm is $15\%$ lower than that of \textsc{CLMNPT21} and $22\%$ lower than that of \textsc{CKLPU24}.
Besides, in Figure~\ref{fig:results-realworld}(\subref{fig:email-vary-d}), our algorithm reduces the clustering cost by up to $14\%$ compared to \textsc{CLMNPT21}. Even in case of poor predictions, \cref{alg:dynamic-stream} does not perform worse than its $(3+\varepsilon)$-approximation counterpart without predictions.

\begin{table}[t!]
\centering
\caption{Clustering costs ($\times$ 1e3) of \cref{alg:dynamic-stream} with binary classifiers as predictors, compared to its non-learning counterpart. The reported values are averaged over $5$ runs.}
\label{tab:binary-classifier}
\begin{tabular}{crrrr}
\toprule
\diagbox{\textbf{Algorithm}}{\textbf{Dataset}} & \makecell{SBM\\($n=1200$)} & \makecell{SBM\\($n=2400$)} & \makecell{SBM\\($n=3600$)} & \textsc{DBLP}\\ \midrule
\textsc{CKLPU24}    & \numprint{105.3}  & \numprint{524.8}  & \numprint{1114.3}  & \numprint{537.5}  \\
\textbf{\cref{alg:dynamic-stream}} & \numprint{35.9} &  \numprint{145.6} & \numprint{324.9}  &  \numprint{247.1} \\ 
\bottomrule
\end{tabular}
\end{table}

\textbf{Results based on binary classifiers as predictors.}
\cref{tab:binary-classifier} shows the performance of \cref{alg:dynamic-stream} with binary classifiers as predictors, a more realistic setting. 
These experiments are performed on three SBM graphs with parameter $p=0.95$ and varying sizes, as well as the \textsc{DBLP} dataset (a sampled subgraph with \numprint{10000} vertices).
The results show that our learning-augmented algorithm consistently outperforms its non-learning counterpart across all datasets. For instance, on the SBM graph with \numprint{2400} vertices, \cref{alg:dynamic-stream} reduces the clustering cost by $72\%$ compared to \textsc{CKLPU24}.

\section*{Acknowledgments}
    The work of YD, SJ, and PP is supported in part by NSFC Grant 62272431 and the Innovation Program for Quantum Science and Technology (Grant No. 2021ZD0302901).
    The work of SL is supported by the State Key Laboratory for Novel Software Technology and the New Cornerstone Science Foundation.

\bibliographystyle{plain}
\bibliography{main}

\appendix

\section{Additional technical preliminaries}
\label{sec:tools}
In this paper, we frequently use graph sparsification techniques, so we review the relevant definitions in this section.
\paragraph{$\ell_0$-samplers.} We first review the definition of $\ell_0$-samplers.
\begin{definition}[$\ell_0$-sampler~\citep{JST11}]
\label{def:l0-sampler}
Let $\vx\in \sR^n$ be a non-zero vector and $\delta\in (0,1)$. 
An $\ell_0$-sampler for $\vx$ returns FAIL with probability at most $\delta$ and otherwise returns some index $i$ such that $x_i \neq 0$ with probability $\frac{1}{|\operatorname{supp}(\vx)|}$ where $\operatorname{supp}(\vx)=\{i\mid x_i\neq 0\}$ is the support of $\vx$.
\end{definition}

The following theorem states that $\ell_0$-samplers can be maintained using a single pass in dynamic streams.
\begin{theorem}[\cite{JST11}]
\label{thm:l0-sampler}
There exists a single-pass streaming algorithm for maintaining an $\ell_0$-sampler for a non-zero vector $\vx\in \sR^n$ (with failure pribability $\delta$) in the dynamic model using $O(\log^2 n \log \delta^{-1})$ bits of space.
\end{theorem}

\paragraph{Cut sparsifiers.} Cut sparsifiers~\cite{BK96} are a basic notion in graph sparsification.
\begin{definition}[Cut sparsifier~\cite{BK96}]
\label{def:cut-sparsifier}
Let $H=(V_H,E_H,w)$ be an undirected graph and let $\varepsilon \in (0,1)$. We say that a reweighted subgraph $H'=(V_H,E_H',w')$ is an \emph{$\varepsilon$-cut sparsifier} of $H$ if for any $A\subseteq V_H$,
\begin{equation*}
(1-\varepsilon)\partial_H(A)\leq \partial_{H'}(A) \leq (1+\varepsilon)\partial_H(A),
\end{equation*}
where $\partial_{H}(A):= \sum_{(u,v)\in E_H:u\in A,v\notin A} w_{uv}$  and $\partial_{H'}(A):= \sum_{(u,v)\in E'_H:u\in A,v\notin A} w'_{uv}$ denote the weights of the cut $(A,V_H\setminus A)$ in $H$ and $H'$, respectively.
\end{definition}

\paragraph{Spectral sparsifiers.} Spectral sparsifiers~\cite{ST11} are a stronger notion than cut sparsifers.
\begin{definition}[Spectral sparsifier~\cite{ST11}]
\label{def:spectral-sparsifier}
Let $H=(V_H,E_H,w)$ be an undirected graph and let $\varepsilon \in (0,1)$. We say that a reweighted subgraph $H'=(V_H,E_H',w')$ is an \emph{$\varepsilon$-spectral sparsifier} of $H$ if for any $\vx\in \mathbb{R}^n$,
\begin{equation*}
(1-\varepsilon)\vx^\top  L_H\vx \leq \vx^\top  L_{H'}\vx \leq (1+\varepsilon)\vx^\top  L_H\vx,
\end{equation*}
which is equivalent to
\begin{equation*}
(1-\varepsilon)L_H \preceq L_{H'} \preceq (1+\varepsilon)  L_H,
\end{equation*}
where $L_H$ and $L_{H'}$ denote the Laplacian matrix of $H$ and $H'$, respectively.
\end{definition}

It is easy to see that if $H'$ is an $\varepsilon$-spectral sparsifier of $H$, then $H'$ is also an $\varepsilon$-cut sparsifier of $H$.

The following theorem states that an $\varepsilon$-spectral sparsifier can be constructed using a single pass and ${O}(\varepsilon^{-2}n\log n)$ space in insertion-only streams.
\begin{theorem}[\cite{KL11}]
\label{thm:insertion-sparsification}
There exists a single-pass streaming algorithm for constructing an $\varepsilon$-spectral sparsifier of an  undirected graph in insertion-only streams using ${O}(\varepsilon^{-2}n\log n)$ space. The algorithm succeeds with high probability.
\end{theorem} 

To construct spectral sparsifiers in dynamic streams, we use the following theorem.
\begin{theorem}[\cite{KMMMNST20fast}]
\label{thm:dynamic-spectral}
There exists a single-pass streaming algorithm for constructing an $\varepsilon$-spectral sparsifier of an  undirected unweighted graph in dynamic streams using $\tilde{O}(\varepsilon^{-2}n)$ space. The algorithm succeeds with high probability.
\end{theorem} 

\paragraph{Effective resistances.} Finally, we review the definition of effective resistances, which are used during the construction of spectral sparsifiers.
\begin{definition}[Effective resistance]
\label{def:eff-res}
Given an undirected graph $G=(V,E,w)$ and a pair of vertices $u,v \in V$, the \emph{effective resistance} between $u$ and $v$ in $G$ is defined as
\begin{align*}
    R_G(u,v):=\max_{\vx \in \sR^{|V|},\vx \neq \vzero} \frac{(x_u - x_v)^2}{\vx^\top L_G \vx},
\end{align*}
where $L_G$ is the Laplacian matrix of $G$.
\end{definition}

We will use the following lower bound on the effective resistance.
\begin{lemma}[\cite{lovasz1993random}]
\label{lem:prop-eff}
    Let $G=(V,E)$ be an undirected unweighted graph. For any $u,v \in V$, we have
    \begin{align*}
        R_G(u,v)\ge \frac{1}{2}\left(\frac{1}{\deg(u)}+\frac{1}{\deg(v)}\right).
    \end{align*}
\end{lemma}

\section{Omitted pseudocodes of \cref{sec:alg-dynamic}}
\label{sec:omitted-pseudocodes-dynamic}
In this section, we give the omitted pseudocodes of \cref{sec:alg-dynamic}: \cref{alg:parallel-pivot}, \cref{alg:pivot}, \cref{alg:CMSY-pivot}, \cref{alg:3-approx-dynamic} and \cref{alg:2.06-approx-dynamic}.

\begin{algorithm}[t]
\begin{algorithmic}[1]
\Require Graph $G^+ =(V,E^+)$, oracle access to a $\beta$-level predictor $\Pi$
\Ensure Clustering/Partition of $V$ into disjoint sets
\State Pick a random permutation of vertices $\pi: V \rightarrow \{1,\dots,n\}$.
\State Initially, all vertices are unclustered and interesting.
\State A vertex $u$ marks itself uninteresting if $\pi_u \geq \tau_u$ where $\tau_u := \frac{c}{\varepsilon}\cdot\frac{n\log n}{\deg^+(u)}$. Here, $\varepsilon \in (0,1/4)$ and $c$ is a universal large constant.
\State \label{line:G_store} Let $G_\textnormal{store}$ be the graph induced by the interesting vertices.
\State \phantomsection \label{line:call-ACN-pivot} $\mathcal{C}_1 \leftarrow \textsc{TruncatedPivot}(G^+, G_\textnormal{store},\pi)$ 
\State \phantomsection \label{line:call-CMSY-pivot} $\mathcal{C}_2 \leftarrow \textsc{TruncatedPivotWithPred}(G^+, G_\textnormal{store},\pi, \Pi)$ 
\State $i \leftarrow \arg\min_{i=1,2} \{\operatorname{cost}_G(\mathcal{C}_i)\}$
\State \Return $\mathcal{C}_i$

\end{algorithmic}
\caption{Offline version of \cref{alg:dynamic-stream}}
\label{alg:parallel-pivot}
\end{algorithm}

\begin{algorithm}[t]
\begin{algorithmic}[1]
\Require Graph $G^+=(V, E^+)$, induced subgraph $H=(V_H, E_H)$ where $V_H\subseteq V$ and $E_H\subseteq E^+$, permutation $\pi: V \rightarrow \{1,\dots,n\}$
\Ensure Clustering/Partition of $V$ into disjoint sets
\State Let $U^{(1)} \leftarrow V_H$ be the set of unclustered vertices in $V_H$.
\State Let $t \leftarrow 1$.
\While{$U^{(t)} \neq \emptyset$}
\State Let $u\in U^{(t)}$ be the vertex with the smallest rank.
\State Mark $u$ as a pivot. Initialize a new \emph{pivot cluster} $S^{(t)} \leftarrow \{u\}$.

\State For each vertex $v\in U^{(t)}$ such that $(u,v)\in E_H$, add $v$ to $S^{(t)}$.
\State Remove all vertices clustered at this iteration from $U^{(t)}$. 
\State $t\leftarrow t+1$.
\EndWhile
\State Each vertex $u\in V\setminus V_H$ joins the cluster of pivot $v$ with the smallest rank, if $(u,v)\in E^+$ and $\pi_v < \tau_u$.
\State Each unclustered vertex $u\in V$ creates a \emph{singleton cluster}.
\State \textbf{return} the final clustering $\mathcal{C}$, which contains all pivot clusters and singleton clusters
\end{algorithmic}
\caption{$\textsc{TruncatedPivot}(G^+, H,\pi)$}
\label{alg:pivot}
\end{algorithm}

\begin{algorithm}[t]
\begin{algorithmic}[1]
\Require Graph $G^+=(V, E^+)$, induced subgraph $H=(V_H, E_H)$ where $V_H\subseteq V$ and $E_H\subseteq E^+$, permutation $\pi: V \rightarrow \{1,\dots,n\}$, oracle access to a $\beta$-level predictor $\Pi$
\Ensure Clustering/Partition of $V$ into disjoint sets
\State Let $U^{(1)} \leftarrow V_H$ be the set of unclustered vertices in $V_H$.
\State Let $t \leftarrow 1$.
\State For any $u,v\in V$, $d_{uv} = \Pi(u,v)$.
\State For any $u,v\in V$, define $p_{uv}:=f(d_{uv})$.
\While{$U^{(t)} \neq \emptyset$}
\State Let $u\in U^{(t)}$ be the vertex with the smallest rank.
\State Mark $u$ as a pivot. Initialize a new \emph{pivot cluster} $S^{(t)} \leftarrow \{u\}$.

\State For each vertex $v\in U^{(t)}$, add $v$ to $S^{(t)}$ with probability $(1-p_{uv})$ independently.
\State Remove all vertices clustered at this iteration from $U^{(t)}$. 
\State $t\leftarrow t+1$.
\EndWhile
\State \label{line:assign} Each vertex $u\in V\setminus V_H$ joins the cluster of pivot $v$ in the order of $\pi$ with probability $(1-p_{uv})$ independently, if $\pi_v < \tau_u$. 
\State Each unclustered vertex $u\in V$ creates a \emph{singleton cluster}.
\State \textbf{return} the final clustering $\mathcal{C}$, which contains all pivot clusters and singleton clusters
\end{algorithmic}
\caption{$\textsc{TruncatedPivotWithPred}(G^+, H,\pi,\Pi)$}
\label{alg:CMSY-pivot}
\end{algorithm}

\begin{algorithm}[t]
\begin{algorithmic}[1]
\Require Graph $G^+=(V,E^+)$
\Ensure Clustering/Partition of $V$ into disjoint sets

\State Pick a random permutation of vertices $\pi: V \rightarrow \{1,\dots,n\}$.
\State Let $U^{(1)} \leftarrow V$ be the set of unclustered vertices.  
\For{$t=1,\dots,n$}
\State Let $\ell_t \leftarrow \frac{c}{\varepsilon}\cdot \frac{n\log n}{t}$. 
\State Let $u\in V$ be the $t$-th vertex in $\pi$ (i.e., $t=\pi_u$).
\State Each unclustered vertex $v$ with $\deg^+(v) \geq \ell_t$ creates a \emph{singleton cluster}.
\If{$u \in U^{(t)}$}
\State Mark $u$ as a pivot. Initialize a new \emph{pivot cluster} $S^{(t)} \leftarrow \{u\}$.

\State For each vertex $v\in N^+(u)\cap U^{(t)}$, add $v$ to $S^{(t)}$.
\EndIf
\State Remove all vertices clustered at this iteration from $U^{(t)}$. 
\EndFor
\State \Return the final clustering $\mathcal{C}$, which contains all pivot clusters and singleton clusters
\end{algorithmic}
\caption{\textsc{CKLPU-Pivot}($G^+$)
}
\label{alg:3-approx-dynamic}
\end{algorithm}

\begin{algorithm}[t]
\begin{algorithmic}[1]
\Require Graph $G^+=(V,E^+)$, oracle access to a $\beta$-level predictor $\Pi$
\Ensure Clustering/Partition of $V$ into disjoint sets

\State Pick a random permutation of vertices $\pi: V \rightarrow \{1,\dots,n\}$.
\State For any $u,v\in V$, $d_{uv} = \Pi(u,v)$.
\State For any $u,v\in V$, define $p_{uv}:=f(d_{uv})$.
\State Let $U^{(1)} \leftarrow V$ be the set of unclustered vertices.  
\For{$t=1,\dots,n$}
\State Let $\ell_t \leftarrow \frac{c}{\varepsilon}\cdot \frac{n\log n}{t}$. 
\State Let $u\in V$ be the $t$-th vertex in $\pi$ (i.e., $t=\pi_u$).
\State \label{line:exclude} Each unclustered vertex $v$ with $\deg^+(v) \geq \ell_t$ creates a \emph{singleton cluster}.
\If{$u \in U^{(t)}$}
\State Mark $u$ as a pivot. Initialize a new \emph{pivot cluster} $S^{(t)} \leftarrow \{u\}$.

\State For each vertex $v\in U^{(t)}$, add $v$ to $S^{(t)}$ with probability $(1-p_{uv})$ independently.
\EndIf
\State Remove all vertices clustered at this iteration from $U^{(t)}$. 
\EndFor
\State \Return the final clustering $\mathcal{C}$, which contains all pivot clusters and singleton clusters
\end{algorithmic}
\caption{\textsc{PairwiseDiss}($G^+, \Pi$)}
\label{alg:2.06-approx-dynamic}
\end{algorithm}

\section{Omitted proofs of \cref{sec:alg-dynamic}}
\label{sec:omitted-details-dynamic}
\subsection{Proof of \cref{thm:main-result-dynamic}}
\label{sec:proof-main-dynamic}
\textbf{Space complexity.} We first analyze the space complexity of \cref{alg:dynamic-stream}. For each vertex $u \in V$, we mainly store its rank $\pi_u$, positive degree $\deg^+(u)$, and $10c\log n\cdot \sigma_u$ independent $\ell_0$-samplers. We have the following lemma which states the space requirement of $\ell_0$-samplers.
\begin{lemma}[\cite{CKLPU24}]
\label{lem:l0-sampler-space}
The $\ell_0$-samplers used in \cref{alg:dynamic-stream} require $O(\varepsilon^{-1}n\log^4 n)$ words of space.
\end{lemma}
Furthermore, by \cref{thm:dynamic-spectral}, the maintenance of an $\varepsilon$-spectral sparsifier in dynamic streams requires $\tilde{O}(\varepsilon^{-2}n)$ words of space. 
Therefore, the space complexity of \cref{alg:dynamic-stream} is $\tilde{O}(\varepsilon^{-2}n)$ words. 

\textbf{Approximation guarantee.} Next, we analyze the approximation ratio of \cref{alg:dynamic-stream}. We rely on the following lemma.
\begin{lemma}[Lemma~2 in \cite{CKLPU24}]
\label{lem:recover}
The $\ell_0$-samplers allow us to recover the positive edges incident to all interesting vertices with high probability.
\end{lemma}
Therefore, \cref{alg:dynamic-stream} works with the same set of edges as \cref{alg:parallel-pivot} in the clustering phase with high probability.
This implies that both algorithms return the same clustering with the same probability.
On the other hand,
if the high probability event of \cref{lem:recover} does not happen, then \cref{alg:dynamic-stream} produces a clustering of cost at most $O(n^2)$, which leads to an additive $1/\operatorname{poly}(n)$ term to the expected cost of \cref{alg:dynamic-stream} compared to that of \cref{alg:parallel-pivot}. This preserves the approximation ratio if $\mathrm{OPT} \neq 0$. 

We also need the following lemma which shows that the estimate
$\widetilde{\operatorname{cost}}_G(\mathcal{C})$ well approximates the cost of any clustering $\mathcal{C}$ of $G$. 
\begin{lemma}[\cite{BCMT23}]
\label{lem:estimated-cost}
Let $\varepsilon\in (0,1)$. For any clustering $\mathcal{C}$ of $V$, the cost $\operatorname{cost}_G(\mathcal{C})$ is approximated by the estimate
$\widetilde{\operatorname{cost}}_G(\mathcal{C}) := \sum_{C \in \mathcal{C}} \left(\frac{1}{2}\partial_{H^+}(C)+\tbinom{|C|}{2} -\frac{1}{2} \sum_{u\in C}\deg^+(u)\right)$ up to a multiplicative factor of $(1\pm \varepsilon)$ with high probability.
\end{lemma}

Therefore, \cref{thm:main-result-dynamic} follows from \cref{lem:parallel}, \cref{lem:recover} and \cref{lem:estimated-cost} by applying the union bound.

\subsection{Proof of \cref{lem:eq-2.06-dynamic}}
The proof is similar to that of \cref{lem:eq-3-dynamic}. The proof idea is as follows: we first show that in both cases, the singleton clusters $V_\textnormal{sin}$ are the same (with the same probability). Then we show that the randomized pivot-based algorithm runs on the same subgraph $G^+[V\setminus V_\textnormal{sin}]$ (with the same probability) in both cases, therefore outputting the same pivot clusters (with the same probability).

    Consider a vertex $u$ that is unclustered at the beginning of iteration $t$ $(\leq \pi_u)$, and becomes a singleton cluster due to \cref{line:exclude} of Algorithm~\textsc{PairwiseDiss}.
    By definition, $t$ is the smallest integer such that $\deg^+(u) \geq \frac{c}{\varepsilon}\cdot \frac{n\log n}{t}$ and hence $t = \lceil \tau_u \rceil$. Since $t \le \pi_u$, we have $\deg^+(u) \geq \frac{c}{\varepsilon}\cdot \frac{n\log n}{\pi_u}$, which corresponds to $u$ becoming uninteresting in \cref{alg:parallel-pivot}. 
    Since $u$ is in a singleton cluster, it did not join any pivot cluster, implying that for any vertex $v \neq u$, either (1) $\pi_v \geq t$, or (2) the event that $v$ is a pivot and $u$ joins the cluster of $v$ satisfying $\pi_v < t$ does not happen. 
    This is equivalent to saying that the event that $u$ joins the cluster of pivot $v$ satisfying 
    $\pi_v < \tau_u$ does not happen, since $\pi_v$ is an integer. By \cref{line:assign} of Algorithm~\textsc{TruncatedPivotWithPred}, $u$ creates a singleton cluster in \cref{line:call-CMSY-pivot} of \cref{alg:parallel-pivot} (with the same probability) as well. 
    
    Now consider a vertex $u$ that creates a singleton cluster in \cref{line:call-CMSY-pivot} of \cref{alg:parallel-pivot}. 
    Then $u$ must be marked uninteresting (implying $\pi_u \geq \tau_u$), and $u$ can neither be a pivot nor join the cluster of pivot $v$ satisfying $\pi_v < \tau_u$. By definition of $\tau_u$, iteration $\lceil\tau_u\rceil$ is the smallest iteration such that $\deg^+(u) \geq \frac{c}{\varepsilon}\cdot \frac{n\log n}{\lceil\tau_u\rceil}$. 
    This implies that $u$ is unclustered at the beginning of iteration $\lceil\tau_u\rceil$ in Algorithm~\textsc{PairwiseDiss}, and forms a singleton cluster in that iteration (with the same probability).
    
    Since the vertices forming singleton clusters are the same in both cases (with the same probability), the subgraph induced by the remaining vertices $G^+[V\setminus V_\textnormal{sin}]$ is the same (with the same probability). 
    The same randomized pivot-based algorithm runs on $G^+[V\setminus V_\textnormal{sin}]$ in both cases, which implies that the pivots will be the same (with the same probability).
    Finally, we observe that in both cases, a non-pivot vertex $u$ joins the cluster of pivot $v$ such that $\pi_v < \tau_u$ in the order of $\pi$ with probability $(1-p_{uv})$ independently.
    Hence, the pivot clusters are the same (with the same probability).

\subsection{Proof of \cref{cor:dynamic-cost-2.06}}
The analysis in \citep{CKLPU24} shows that the cost of good edges can be charged to the pivot clusters. The following lemma shows that the cost of bad edges can also be related to the pivot clusters, allowing us to bound the overall clustering cost.
\begin{lemma}[\cite{CKLPU24}]
\label{lem:seq-cost}
    Let $\varepsilon\in (0,1/4)$.
    Let $P$ denote the cost of pivot clusters, and let $W$ denote the cost of the final clustering, then $\E[W]=\E[P+|E_\textnormal{bad}|]\leq (1+4\varepsilon)\E[P] + \frac{1+4\varepsilon}{n^{\alpha-2}}$,
    where $\alpha := c/2-1 \gg 2$.
\end{lemma}

Then \cref{cor:dynamic-cost-2.06} follows from \cref{lem:eq-2.06-dynamic}, \cref{lem:seq-cost} and \cref{lem:seq-pairwisediss-pivot}.

\subsection{Proof of \cref{lem:parallel}}
    \cref{lem:parallel} follows from \cref{lem:dynamic-cost-3} and \cref{cor:dynamic-cost-2.06}.
    Note that in \cref{lem:dynamic-cost-3}, we can substitute $\varepsilon' := 12\varepsilon$, where $\varepsilon$ can be arbitrarily small.
    If $\mathrm{OPT}\geq 1$, then $\E[\mathrm{cost}_G(\mathcal{C}_1)]\leq (3+12\varepsilon)\cdot \mathrm{OPT}$, which gives a $(3+\varepsilon')$-approximation in expectation. 
    If $\mathrm{OPT} = 0$, then $\E[\mathrm{cost}_G(\mathcal{C}_1)]=1/\operatorname{poly}(n)$.
    Similarly, in \cref{cor:dynamic-cost-2.06}, we can substitute $\varepsilon' := 8.24\beta\varepsilon$.

\section{Our algorithm for complete graphs in insertion-only streams}
\label{sec:omitted-details-insertion}
In this section, we propose an algorithm for complete graphs in insertion-only streams. This algorithm is different from our algorithm for dynamic streams (\cref{alg:dynamic-stream}), while achieving the same approximation guarantee with improved space complexity.
It is also simpler and more practical than the existing $1.847$-approximation algorithm~\cite{CLPTYZ24}, which is based on local search and requires enumerating a large number of subsets of a constant-size set, making it quite impractical. 

\begin{theorem}
\label{thm:main-result-insertion} 
Let $\varepsilon\in (0,1)$ and $\beta \geq 1$. Given oracle access to a $\beta$-level predictor, there exists a single-pass streaming algorithm that, with high probability, achieves an expected $(\min\{2.06\beta, 3\}+\varepsilon)$-approximation for Correlation Clustering on complete graphs in insertion-only streams. The algorithm uses ${O}(\varepsilon^{-2}n\log n)$ words of space. 
\end{theorem}

\subsection{Overview}
We first briefly describe a single-pass $(3+\varepsilon)$-approximation streaming algorithm by Chakrabarty and Makarychev~\cite{CM23}. Initially, the algorithm adds a positive self-loop for each vertex and picks a random ordering $\pi:V \rightarrow \{1,\dots,n\}$. The rank of $u$ is denoted as $\pi_u$. Then it scans the input stream. For each vertex, the algorithm stores its at most $k$ positive neighbors with lowest ranks, where $k$ is a constant. Subsequently, it runs the \textsc{Pivot} algorithm \citep{ACN08} on the stored graph, where it picks pivots in the order of $\pi$. Finally, it puts unclustered vertices in singleton clusters.

Our main idea is to incorporate the above algorithm with the algorithm proposed by Chawla et al.~\cite{CMSY15}.  
Specifically, our algorithm employs two different methods to store at most $k$ neighbors of each vertex. The first method is the same as the algorithm proposed by Chakrabarty and Makarychev~\cite{CM23} and the second method is adapted from the work of Chawla et al.~\cite{CMSY15}, which adds neighbors with probabilities determined by predictions of pairwise distances. Finally, we obtain two clusterings (denoted as $\mathcal{C}_1$ and $\mathcal{C}_2$) and output the one with the lower cost. Similar to \cref{alg:dynamic-stream}, here we also need to use the graph sparsification technique \citep{KL11} to approximate the cost of a clustering.

\subsection{The algorithm}
Recall that we have oracle access to a  $\beta$-level predictor $\Pi$, which can predict the pairwise distance $d_{uv} \in [0,1]$ between any two vertices $u$ and $v$ in $G$.

Based on the predictions, we propose a single-pass semi-streaming algorithm which works in insertion-only streams. The pseudocode is given in \cref{alg:insertion-only}. 
We first pick a random permutation of vertices $\pi: V \rightarrow \{1,\dots,n\}$. For each vertex $u \in V$, we initialize two priority queues $A(u)$ and $B(u)$, each with a maximum size capped at $k$, where $k$ is a constant. Initially, we add $u$ to both queues.
During the streaming phase, we employ two distinct methods to retain at most $k$ neighbors of each vertex. 
Specifically, for each edge $(u,v)\in E$ in the stream, if $(u,v)$ is a positive edge, we add $u$ to $A(v)$ and add $v$ to $A(u)$. Additionally, regardless of whether $(u,v)$ is positive or negative, we add $u$ to $B(v)$ with probability $(1-p_{uv})$ and add $v$ to $B(u)$ with probability $(1-p_{uv})$, where $p_{uv}=f(d_{uv})$ and $d_{uv}= \Pi(u,v)$. 
Note that if the size of any queue exceeds $k$, then we remove the vertex with the highest rank from the queue. That is, $A(u)$ maintains at most $k$ positive neighbors of $u$ with lowest ranks, while $B(u)$ contains at most $k$ neighbors (not necessarily positive) of $u$ with lowest ranks, the inclusion of which is probabilistic. Note that we define the rank of a vertex as its order in the permutation $\pi$, e.g., $\pi_u$ is the rank of $u$.

After the streaming phase, we run \cref{alg:clustering} on the truncated graphs induced by both sets of priority queues, i.e., $\{A(u)\}_{u\in V}$ and $\{B(u)\}_{u\in V}$. 
Specifically, for each vertex $u$ picked in the order of $\pi$, we determine the cluster to which $u$ belongs. We try to find the vertex $v$ with the lowest rank in the queue of $u$, such that $v$ is a pivot or $v=u$. If such a vertex $v$ does not exist, then we mark $u$ as a singleton and place it in a singleton cluster. Otherwise,
we assign $u$ to the cluster of $v$. In particular, if $v=u$, then we mark $u$ as a pivot.
Finally, we obtain two clusterings, each corresponding to a set of priority queues.
We output the clustering with the lower cost. 

\begin{algorithm}[t]
\begin{algorithmic}[1]
\Require Complete graph $G=(V,E=E^+ \cup E^-)$ as an arbitrary-order stream of edges, oracle access to a $\beta$-level predictor $\Pi$, integer $k$
\Ensure Clustering/Partition of $V$ into disjoint sets
\LeftComment{\textbf{Pre-processing phase}}
\State Pick a random permutation of vertices $\pi: V \rightarrow \{1,\dots,n\}$.
\State For any $u,v\in V$, $d_{uv} = \Pi(u,v)$.
\State For any $u,v\in V$, define $p_{uv}:=f(d_{uv})$.
\For{each vertex $u\in V$}
\State Create a priority queue $A(u)$ with a maximum size of $k$ and initialize $A(u) \leftarrow \{u\}$.
\State Create a priority queue $B(u)$ with a maximum size of $k$ and initialize $B(u) \leftarrow \{u\}$.
\State $\deg^+(u) \leftarrow 0$
\EndFor
\LeftComment{\textbf{Streaming phase}}
\For{each edge $e=(u,v)\in E$} 
\If{$e=(u,v)\in E^+$}
\State Add $u$ to $A(v)$. Add $v$ to $A(u)$.
\If{$|A(u)|> k$ (resp. $|A(v)|> k$)}
\State Remove the vertex with the highest rank from $A(u)$ (resp. $A(v)$).
\EndIf
\State $\deg^+(u) \leftarrow \deg^+(u) + 1, \deg^+(v) \leftarrow \deg^+(v) + 1$
\EndIf
\State With probability $(1-p_{uv})$, add $u$ to $B(v)$ and add $v$ to $B(u)$.

\If{$|B(u)|> k$ (resp. $|B(v)|> k$)}
\State Remove the vertex with the highest rank from $B(u)$ (resp. $B(v)$).
\EndIf
\EndFor
\State Maintain an $\varepsilon$-spectral sparsifier $H^+$ for $G^+$ using the algorithm of \cref{thm:insertion-sparsification}.
\LeftComment{\textbf{Post-processing phase}}
\State \label{line:cluster-A} $\mathcal{C}_1 \leftarrow \textsc{Cluster}(V,\pi,\{A(u)\}_{u\in V})$ 
\State \label{line:cluster-B} $\mathcal{C}_2 \leftarrow \textsc{Cluster}(V,\pi,\{B(u)\}_{u\in V})$ 
\State $\widetilde{\operatorname{cost}}_G(\mathcal{C}_1) \leftarrow \sum_{C \in \mathcal{C}_1} (\frac{1}{2}\partial_{H^+}(C)+\tbinom{|C|}{2} -\frac{1}{2} \sum_{u\in C}\deg^+(u) )$
\State $\widetilde{\operatorname{cost}}_G(\mathcal{C}_2) \leftarrow \sum_{C \in \mathcal{C}_2} (\frac{1}{2}\partial_{H^+}(C)+\tbinom{|C|}{2} -\frac{1}{2} \sum_{u\in C}\deg^+(u) )$
\State $i \leftarrow \arg\min_{i=1,2} \{\widetilde{\operatorname{cost}}_G(\mathcal{C}_i)\}$.
\State \Return $\mathcal{C}_i$

\end{algorithmic}
\caption{An algorithm for complete graphs in insertion-only streams}
\label{alg:insertion-only}
\end{algorithm}

\begin{algorithm}[t]
\begin{algorithmic}[1]
\Require Vertex set $V$, permutation of vertices $\pi: V \rightarrow \{1,\dots,n\}$, truncated neighbors of each vertex $\{T(u)\}_{u\in V}$
\Ensure Clustering/Partition of $V$ into disjoint sets
\For{each unclustered vertex $u\in V$ chosen in the order of $\pi$}
\State Find the vertex $v\in T(u)$ with the lowest rank such that $v$ is a pivot or $v=u$, i.e., 
$v \leftarrow \arg\min_{v\in T(u)}\{ \pi_v: v\text{ is a pivot} \text{ or } v=u \}$.
\If{such a vertex $v$ exists}
\State Put $u$ in the cluster of $v$. 
\If{$v=u$}
\State Mark $u$ as a \emph{pivot}.
\EndIf

\Else
\State Put $u$ in a singleton cluster. Mark $u$ as a \emph{singleton}.
\EndIf
\EndFor
\State \textbf{return} the final clustering $\mathcal{C}$
\end{algorithmic}
\caption{$\textsc{Cluster}(V,\pi,\{T(u)\}_{u\in V})$}
\label{alg:clustering}
\end{algorithm}

It is worth noting that in the final step, the cost of a clustering cannot be exactly calculated, as our streaming algorithm cannot store the entire graph. To overcome this challenge, we borrow the idea from the work of Behnezhad et al.~\cite{BCMT23} and utilize the graph sparsification technique \citep{KL11} to estimate the the cost of a clustering. 
Specifically, during the streaming phase, we maintain an $\varepsilon$-spectral sparsifier $H^+$ for $G^+$ using the algorithm of \cref{thm:insertion-sparsification}. 
We also maintain the positive degree $\deg^+(u)$ for each vertex $u$. Then we can approximate the cost of a clustering up to a $(1\pm \varepsilon)$-multiplicative error with high probability, by the guarantee of \cref{lem:estimated-cost}.

\subsection{Analysis of \cref{alg:insertion-only}}

\textbf{Space complexity.} 
For each vertex $u \in V$, we mainly store its rank $\pi_u$, positive degree $\deg^+(u)$, and at most $2k$ vertices. As we will see, we set $k=O(1/\varepsilon)$.
Furthermore, by \cref{thm:insertion-sparsification}, the maintenance of an $\varepsilon$-spectral sparsifier in insertion-only streams requires ${O}(\varepsilon^{-2}n\log n)$ words of space.
Therefore, the total space complexity of \cref{alg:insertion-only} is $O(\varepsilon^{-2}n\log n)$ words.

\textbf{Correctness.} 
Since the final clustering returned by our algorithm is the one with the lower cost between the two on the truncated graphs, we begin by analyzing their costs.
Similar to the analysis of \cref{alg:dynamic-stream}, for ease of analysis, we separately examine the approximation ratios of the corresponding offline versions (Algorithms \textsc{CM-Pivot} and \textsc{PairwiseDiss2}) that equivalently output these two clusterings.

\textbf{Algorithm \textsc{CM-Pivot} \textnormal{\citep{CM23}}.} This algorithm proceeds in iterations. Let $F^{(t)}$ denote the set of fresh vertices and $U^{(t)}$ denote the set of unclustered vertices at the beginning of iteration $t$. Additionally, we maintain a counter $K^{(t)}(u)$ for each vertex $u \in V$.
Initially, all the vertices are fresh and unclustered, with the counters set to $0$. At iteration $t$, we pick a vertex $w^{(t)}$ from the set of fresh vertices $F^{(t)}$ uniformly at random. If $w^{(t)}$ is unclustered, then we mark it as a pivot and create a cluster $S^{(t)}$ containing $w^{(t)}$ and all of its unclustered positive neighbors. Otherwise, we increment the counters for all unclustered positive neighbors of $w^{(t)}$. Subsequently, vertices whose counters reach the value of $k$ are assigned to singleton clusters. At the end of iteration $t$, we remove $w^{(t)}$ from $F^{(t)}$ and remove all vertices clustered in this iteration from $U^{(t)}$. Then the algorithm proceeds to the next iteration. Finally, we output all pivot clusters and singleton clusters.
The pseudocode is given in \cref{alg:3-approx}.

\begin{algorithm}[t]
\begin{algorithmic}[1]
\Require Complete graph $G=(V,E=E^+ \cup E^-)$, integer $k$
\Ensure Partition of vertices into disjoint sets

\State Let $F^{(1)} \leftarrow V$ be the set of fresh vertices. 
\State Let $U^{(1)} \leftarrow V$ be the set of unclustered vertices.  
\State For each vertex $u\in V$, initialize a counter $K^{(1)}(u) \leftarrow 0$.
\State Let $t\leftarrow 1$.
\While{$F^{(t)} \neq \emptyset$}

\State Choose a vertex $w^{(t)} \in F^{(t)}$ uniformly at random. 
\If{$w^{(t)} \in U^{(t)}$}
\State Mark $w^{(t)}$ as a pivot. Initialize a new \emph{pivot cluster} $S^{(t)} \leftarrow \{w^{(t)}\}$.

\State For each vertex $v\in N^+(w^{(t)})\cap U^{(t)}$, add $v$ to $S^{(t)}$.

\Else
\State For each vertex $v\in N^+(w^{(t)})\cap U^{(t)}$, let $K^{(t+1)}(v) \leftarrow K^{(t)}(v)+1$. Subsequently, all vertices $v$ with $K^{(t+1)}(v) = k$ are put into \emph{singleton clusters}.
\EndIf
\State Let $F^{(t+1)} \leftarrow F^{(t)} \setminus \{w^{(t)}\}$ and remove all vertices clustered at this iteration from $U^{(t)}$. 
\State Let $t \leftarrow t+1$.
\EndWhile
\State \Return the final clustering $\mathcal{C}$, which contains all pivot clusters and singleton clusters
\end{algorithmic}
\caption{\textsc{CM-Pivot}($G,k$)
}
\label{alg:3-approx}
\end{algorithm}

\textbf{Algorithm \textsc{PairwiseDiss2}.} 
This algorithm has oracle access to a $\beta$-level predictor $\Pi: \tbinom{V}{2} \rightarrow [0,1]$. This algorithm closely resembles Algorithm \textsc{CM-Pivot}, differing in the following two aspects: (1) If $w^{(t)}\in U^{(t)}$, then we create a cluster $S^{(t)}$ containing $w^{(t)}$ and add all unclustered vertices $v$ to $S^{(t)}$ with probability $(1-p_{vw^{(t)}})$ independently, where $p_{vw^{(t)}}=f(d_{vw^{(t)}})$ and $d_{vw^{(t)}} = \Pi(v,w^{(t)})$. (2) If $w^{(t)}\notin U^{(t)}$, we increment the counters for all unclustered vertices $v$ with probability $(1-p_{vw^{(t)}})$.
The pseudocode is given in \cref{alg:2.06-approx}.

\begin{algorithm}[t]
\begin{algorithmic}[1]
\Require Complete graph $G=(V,E=E^+ \cup E^-)$, oracle access to a $\beta$-level predictor $\Pi$, integer $k$
\Ensure Partition of vertices into disjoint sets

\State Let $F^{(1)} \leftarrow V$ be the set of fresh vertices. 
\State Let $U^{(1)} \leftarrow V$ be the set of unclustered vertices.  
\State For each vertex $u\in V$, initialize a counter $K^{(1)}(u) \leftarrow 0$.
\State For any $u,v\in V$, $d_{uv} = \Pi(u,v)$.
\State For any $u,v\in V$, define $p_{uv}:=f(d_{uv})$.
\label{line:function}
\State Let $t\leftarrow 1$.
\While{$F^{(t)} \neq \emptyset$}

\State Choose a vertex $w^{(t)} \in F^{(t)}$ uniformly at random. 
\If{$w^{(t)} \in U^{(t)}$}
\State Mark $w^{(t)}$ as a pivot. Initialize a new \emph{pivot cluster} $S^{(t)} \leftarrow \{w^{(t)}\}$.

\State For each vertex $v\in U^{(t)}$, add $v$ to $S^{(t)}$ with probability $(1-p_{vw^{(t)}})$ independently.

\Else
\State For each vertex $v\in U^{(t)}$, let $K^{(t+1)}(v) \leftarrow K^{(t)}(v)+1$ with probability $(1-p_{vw^{(t)}})$ independently. Subsequently, all vertices $v$ with $K^{(t+1)}(v) = k$ are put into \emph{singleton clusters}.
\EndIf
\State Let $F^{(t+1)} \leftarrow F^{(t)} \setminus \{w^{(t)}\}$ and remove all vertices clustered at this iteration from $U^{(t)}$. 
\State Let $t \leftarrow t+1$.
\EndWhile
 \State \Return the final clustering $\mathcal{C}$, which contains all pivot clusters and singleton clusters
\end{algorithmic}
\caption{\textsc{PairwiseDiss2}($G,\Pi,k$)}
\label{alg:2.06-approx}
\end{algorithm}
\subsubsection{\cref{alg:insertion-only} as a combination of Algorithms \textsc{CM-Pivot} and \textsc{PairwiseDiss2}}
We define a permutation $\pi$ for Algorithms \textsc{CM-Pivot} and \textsc{PairwiseDiss2} as $\pi: w^{(t)} \mapsto t$, where $w^{(t)}$ is the vertex picked at iteration $t$ of Algorithms \textsc{CM-Pivot} and \textsc{PairwiseDiss2}. Obviously, $\pi$ is a uniformly random permutation over $V$. 
Therefore, we can also view Algorithms \textsc{CM-Pivot} and \textsc{PairwiseDiss2} from an equivalent perspective: at the beginning of each iteration $t$, choose a vertex $w^{(t)}$ in the order of $\pi$.
We have the following lemmas.

\begin{lemma}[Lemma~2.1 in \cite{CM23}]
\label{lem:eq-3}
    If \cref{alg:insertion-only} and Algorithm~\textsc{CM-Pivot} use the same permutation $\pi$, then Algorithm~\textsc{CM-Pivot} and Line~\ref{line:cluster-A} of \cref{alg:insertion-only} output the same clustering of $V$. 
\end{lemma}

\begin{lemma}
\label{lem:eq-2.06}
    If \cref{alg:insertion-only} and Algorithm~\textsc{PairwiseDiss2} use the same permutation $\pi$ and predictions $\{d_{uv}\}_{u,v\in V}$, then Algorithm~\textsc{PairwiseDiss2} and Line~\ref{line:cluster-B} of \cref{alg:insertion-only} output the same clustering of $V$ with the same probability.
\end{lemma}
\begin{proof}
    The proof is similar to that of \cref{lem:eq-3}. 
    Suppose that \cref{alg:insertion-only} and Algorithm~\textsc{PairwiseDiss2} use the same permutation $\pi$ and predictions $\{d_{uv}\}_{u,v\in V}$, we want to prove that for each vertex $u\in V$, with the same probability, in both clusterings returned by Algorithm~\textsc{PairwiseDiss2} and Line~\ref{line:cluster-B} of \cref{alg:insertion-only}, $u$ is either assigned to the same pivot, or $u$ is placed into a singleton cluster.
    
    We prove by induction on the rank $\pi_u$. Suppose that all vertices $v$ with $\pi_v < \pi_u$ are clustered in the same way with the same probability. If $u$ is put into a singleton cluster in the clustering returned by Line~\ref{line:cluster-B} of \cref{alg:insertion-only}, then there must exist $k$ vertices added to $B(u)$ probabilistically, and their ranks are lower than $\pi_u$. None of the vertices in $B(u)$ are pivots. Since both algorithms use the same $\pi$ and $\{d_{uv}\}_{u,v\in V}$, in Algorithm~\textsc{PairwiseDiss2}, these $k$ vertices will cause the counter of $u$ to increment $k$ times probabilistically. Therefore, $u$ is also placed in a singleton cluster in the clustering returned by Algorithm~\textsc{PairwiseDiss2}. And vice versa.

    In \cref{alg:insertion-only}, if there are any pivots in $B(u)$ (or $u$ itself), then $u$ will be assigned to the pivot with the lowest rank (denoted as $v$). We have $\pi_v \leq \pi_u$ and $v$ has been added to $B(u)$ probabilistically. In Algorithm~\textsc{PairwiseDiss2}, with the same probability, $v$ is marked as a pivot and $u$ is added to the cluster of $v$. And vice versa.

    Therefore, Algorithm~\textsc{PairwiseDiss2} and Line~\ref{line:cluster-B} of \cref{alg:insertion-only} cluster $u$ in the same way with the same probability.
\end{proof}

\subsubsection{The approximation ratios of \textsc{CM-Pivot} and \textsc{PairwiseDiss2}}
In order to analyze the approximation ratio of \cref{alg:insertion-only}, it suffices to analyze Algorithms~\textsc{CM-Pivot} and~\textsc{PairwiseDiss2} respectively. 
We follow the analysis framework in the work of Chakrabarty and Makarychev~\cite{CM23}. We categorize all iterations into \emph{pivot iterations} and \emph{singleton iterations}.
Both iterations create some clusters.
Consider iteration $t$ of both algorithms.
If $w^{(t)}\in U^{(t)}$, then iteration $t$ is a pivot iteration; otherwise, it is a singleton iteration. 
We say that an edge $(u,v)$ is \emph{decided} at iteration $t$ if both $u$ and $v$ were not clustered at the beginning of iteration $t$ (i.e., $u,v\in U^{(t)}$) but at least one of them was clustered at iteration $t$. 
Once an edge $(u,v)$ is decided, we can determine whether it contributes to the cost of the algorithm (i.e., the number of disagreements). 
Specifically, if $(u,v)\in E^+$, then it contributes to the cost of the algorithm if exactly one of $u$ and $v$ is assigned to the newly created cluster $S^{(t)}$;
if $(u,v)\in E^-$, then it contributes to the cost of the algorithm if both $u$ and $v$ are assigned to the newly created cluster $S^{(t)}$. 

Let $E^{(t)}$ denote the set of decided edges at pivot iteration $t$. Specifically, $E^{(t)}=\{(u,v)\in E: u,v\in U^{(t)}; u\in S^{(t)} \text{ or } v\in S^{(t)}\}$.
Let $P^{(t)}$ denote the cost of decided edges at pivot iteration $t$. 
We call the clusters created in pivot iterations \emph{pivot clusters}.
Let $P$ denote  
the cost of all pivot clusters. 
Therefore, $P=\sum_{\text{$t$ is a pivot iteration}} P^{(t)}$.
Let $S$ denote the cost of all singleton clusters. 
Therefore, the cost of the algorithm is equal to $P+S$.

\textbf{Analysis of Algorithm~\textsc{CM-Pivot}.}
We have the following guarantees of Algorithm~\textsc{CM-Pivot}.
\begin{lemma}[\cite{CM23}]
\label{lem:offline-cost-pivot-3}
Let $P_1$ denote the cost of pivot clusters returned by Algorithm~\textsc{CM-Pivot}, then $\E[P_1]\leq 3\cdot \mathrm{OPT}$, where $\mathrm{OPT}$ is the cost of the optimal solution on $G$.
\end{lemma}
\begin{lemma}[\cite{CM23}]
\label{lem:offline-cost-singleton-3}
Let $S_1$ denote the cost of singleton clusters returned by Algorithm~\textsc{CM-Pivot}, then $\E[S_1]\leq \frac{6}{k-1}\cdot \mathrm{OPT}$.
\end{lemma}

Therefore, we can bound the cost of the clustering returned by Line~\ref{line:cluster-A} of \cref{alg:insertion-only}.
\begin{lemma}
\label{cor:streaming-cost-3}
Let $P_1$ and $S_1$ denote the costs of pivot clusters and singleton clusters, respectively, returned by Algorithm~\textsc{CM-Pivot}. Let $\mathcal{C}_1$ denote the clustering returned by Line~\ref{line:cluster-A} of \cref{alg:insertion-only}.
Then $\E[\mathrm{cost}_G(\mathcal{C}_1)]=\E[P_1+S_1]\leq ( 3+ \frac{6}{k-1}) \cdot \mathrm{OPT}$.
\end{lemma}
\begin{proof}
    \cref{cor:streaming-cost-3} follows from \cref{lem:eq-3}, \cref{lem:offline-cost-pivot-3} and \cref{lem:offline-cost-singleton-3}.
\end{proof}

\textbf{Analysis of Algorithm~\textsc{PairwiseDiss2}.}
Next, we analyze the approximation ratio of Algorithm~\textsc{PairwiseDiss2}. We first bound the cost of pivot clusters.
\begin{lemma}
\label{lem:pivot-cost}
Let $P_2$ denote the cost of pivot clusters returned by Algorithm~\textsc{PairwiseDiss2},
then $\E[P_2]\leq 2.06\beta \cdot \mathrm{OPT}$.
\end{lemma}
\begin{proof}
The key observation is that the pivot iterations in Algorithm~\textsc{PairwiseDiss2} are equivalent to the iterations of $2.06$-approximation LP rounding algorithm by Chawla et al.~\cite{CMSY15}: given that $w^{(t)}$ is unclustered (i.e., $w^{(t)}\in U^{(t)}$), the conditional distribution of $w^{(t)}$ is uniformly distributed in $U^{(t)}$, and the cluster created during this iteration contains $w^{(t)}$ and all unclustered vertices $v$ added with probability $(1-p_{vw^{(t)}})$.
Therefore, we can directly apply the triangle-based analysis in the work of Chawla et al.~\cite{CMSY15}. Define $L:=\sum_{(u, v) \in E^{+}} d_{u v}+\sum_{(u, v) \in E^{-}}(1-d_{u v})$. Since the predictor is $\beta$-level, by \cref{def:beta-level-predictor}, we have that the predictions $\{d_{uv}\}_{u,v\in V}$ satisfy triangle inequality and $L \leq \beta \cdot \mathrm{OPT}$.
It follows that for all pivot iterations $t$, $\E [P_2^{(t)}]\leq 2.06 \cdot\E[L^{(t)}]$, where $P_2^{(t)}$ is the cost induced by the pivot cluster created at iteration $t$, and
$L^{(t)} := \sum_{(u, v) \in E^{+}\cap E^{(t)}} d_{u v}+\sum_{(u, v) \in E^{-}\cap E^{(t)}}(1-d_{u v})$.
By linearity of expectation, we have
$\E[P_2]=\E[\sum_{\text{$t$ is a pivot iteration}} P_2^{(t)}]
    =\sum_{\text{$t$ is a pivot iteration}} \E[P_2^{(t)}]
    \leq 2.06\cdot L 
    \leq 2.06\beta \cdot \mathrm{OPT}.$ 
\end{proof}

Then we bound the cost of singleton clusters returned by Algorithm~\textsc{PairwiseDiss2}, denoted as $S_2$.
We highlight that this part is non-trivial. 
Different from the analysis in the work of Chakrabarty and Makarychev~\cite{CM23} which designs a potential function and shows that it is a submartingale, we consider an algorithm equivalent to Algorithm~\textsc{PairwiseDiss2}. 
In this algorithm, we construct a random subgraph $G':=(V,E'^+\cup E'^-)$ where each edge $(u,v)\in E$ is added to $E'^+$ with probability $(1-p_{uv})$ and added to $E'^-$ with the remaining probability. Then we perform Algorithm~\textsc{CM-Pivot} on $G'$. 
In other words, we first preround the $\beta$-level predictions $\{d_{uv}\}_{u,v\in V}$ to obtain a new instance $G'$ and then run Algorithm~\textsc{CM-Pivot} on $G'$ where the positive edges are induced by the predictions.
The pseudocode is given in \cref{alg:2.06-approx-prerounding}.

\begin{algorithm}[t]
\begin{algorithmic}[1]
\Require Complete graph $G=(V,E=E^+ \cup E^-)$, oracle access to a $\beta$-level predictor $\Pi$, integer $k$
\Ensure Partition of vertices into disjoint sets

\State For any $u,v\in V$, $d_{uv} = \Pi(u,v)$.
\State For any $u,v\in V$, define $p_{uv}:=f(d_{uv})$.
\State $E'^+ \leftarrow \emptyset$.
\For{each edge $(u,v)\in E$ such that $p_{uv}<1$}
\State add $(u,v)$ to $E'^+$ with probability $(1-p_{uv})$.
\EndFor
\State $E'^- \leftarrow E \setminus E'^+$
\State $\mathcal{C}\leftarrow \textsc{CM-Pivot}(G':=(V,E'^+\cup E'^-),k)$
 \State \Return $\mathcal{C}$
\end{algorithmic}
\caption{\textsc{PairwiseDiss2WithPrerounding}($G,\Pi,k$)
}
\label{alg:2.06-approx-prerounding}
\end{algorithm}

We first show the equivalence of Algorithm~\textsc{PairwiseDiss2} and Algorithm \textsc{PairwiseDiss2WithPrerounding}.
\begin{lemma}
\label{cla:eq-prerounding}
    If Algorithm~\textsc{PairwiseDiss2} and Algorithm~\textsc{PairwiseDiss2WithPrerounding} use the same permutation $\pi$ and predictions $\{d_{uv}\}_{u,v\in V}$, then they produce the same clustering with the same probability.
\end{lemma}
\begin{proof}
    The randomness in both algorithms comes from two sources: (1) the uniformly random permutation $\pi$ on vertices and (2) the probability that each vertex $v$ adjacent to  $w^{(t)}$ will join the cluster of $w^{(t)}$ or increment its counter. The main difference between the two algorithms lies in the order in which the two sources of randomness are revealed:
    Algorithm~\textsc{PairwiseDiss2} can be viewed as choosing $\pi$ at the beginning and then performing iterations, where the randomness of all edges incident to $w^{(t)}$ is revealed after $w^{(t)}$ is chosen. In contrast, Algorithm~\textsc{PairwiseDiss2WithPrerounding} reveals the randomness of edges at the beginning, uses this information to construct a new instance, and then performs Algorithm~\textsc{CM-Pivot} on the new instance, where the randomness for $\pi$ is revealed. Note that the order of randomness does not affect the output. 
    Therefore, if both algorithms use the same $\pi$ and $\{d_{uv}\}_{u,v\in V}$, then they will output the same clustering with the same probability.
\end{proof}

Therefore, we can directly apply \cref{lem:offline-cost-singleton-3} to $G'$. To this end, we first show that
$G'$ still well preserves the Correlation Clustering structure of $G$, by showing that
the optimal solution on $G'$ does not differ from the optimal solution on $G$ by a lot.
\begin{lemma}
\label{lem:opt-opt'}
    $\E[\mathrm{OPT}'] \leq (2\beta+1)\cdot \mathrm{OPT}$, where $\mathrm{OPT}$ and $\mathrm{OPT}'$ are the costs of the optimal solutions on $G$ and $G'$, respectively.
\end{lemma}

\begin{proof}
    Let $\mathcal{C}^*$ be the optimal clustering on $G$ with cost $\mathrm{OPT}$. For any $u,v\in V$, let $x_{uv}^* \in \{0,1\}$ indicate whether $u$ and $v$ are in the same cluster or not in $\mathcal{C}^*$. Specifically, if $u$ and $v$ are in the same cluster in $\mathcal{C}^*$, then $x_{uv}^*=0$; otherwise, $x_{uv}^*=1$.
    Let $\mathcal{C}'^*$ be the optimal clustering on $G'$ with cost $\mathrm{OPT}'$.
    Then we have
    \begin{align*}
        &\E[\mathrm{OPT}'] =\E[\mathrm{cost}_{G'}(\mathcal{C'}^*)]
        \leq \E[\mathrm{cost}_{G'}(\mathcal{C}^*)]\\
        =~&\sum_{(u,v)\in E^+}[x_{uv}^*(1-p_{uv}) + (1-x_{uv}^*) p_{uv}] 
        + \sum_{(u,v)\in E^-} [x_{uv}^*(1-p_{uv}) + (1-x_{uv}^*) p_{uv}] \\
        =~& \sum_{(u,v)\in E^+}x_{uv}^* + \sum_{(u,v)\in E^-} (1-x_{uv}^*) + \sum_{(u,v)\in E^+} p_{uv}(1-2x_{uv}^*)
        + \sum_{(u,v)\in E^-} (1-p_{uv})(2x_{uv}^*-1)\\
        \leq~&  \mathrm{OPT} + \sum_{(u,v)\in E^+} p_{uv} + \sum_{(u,v)\in E^-} (1-p_{uv})\\
        \leq~&  \mathrm{OPT} + \sum_{(u,v)\in E^+} 2d_{uv} + \sum_{(u,v)\in E^-} (1-d_{uv})
        \leq  (1+2\beta)\cdot \mathrm{OPT},
    \end{align*}
    where the first step follows from $\mathrm{cost}_{G'}(\mathcal{C'}^*)=\mathrm{OPT}'$, 
    the second step follows from that $\mathcal{C}'^*$ is the optimal clustering on $G'$, 
    the third step follows from our construction of $G'$, 
    the fifth step follows from $\sum_{(u,v)\in E^+}x_{uv}^* + \sum_{(u,v)\in E^-} (1-x_{uv}^*)=\mathrm{OPT}$ and $\sum_{(u,v)\in E^+} p_{uv}(1-2x_{uv}^*)+\sum_{(u,v)\in E^-}  (1-p_{uv})(2x_{uv}^*-1) \leq \sum_{(u,v)\in E^+} p_{uv} + \sum_{(u,v)\in E^-} (1-p_{uv})$ since $1-2x_{uv}^*\in \{-1,1\}$, 
    the sixth step follows from our choice for $p_{uv}$, 
    and the last step follows from $\sum_{(u,v)\in E^+} 2d_{uv} + \sum_{(u,v)\in E^-} (1-d_{uv}) \leq 2(\sum_{(u,v)\in E^+} d_{uv} + \sum_{(u,v)\in E^-} (1-d_{uv}))\leq 2\beta \cdot \mathrm{OPT}$.
\end{proof}

Now we are ready to bound the cost of singleton clusters and, consequently, the final clustering returned by Algorithm~\textsc{PairwiseDiss2}.
\begin{lemma}
\label{lem:singleton-cost}
Let $S_2$ denote the cost of singleton clusters returned by Algorithm~\textsc{PairwiseDiss2}, then
$\E[S_2]\leq \frac{6(2\beta+1)}{k-1} \cdot \mathrm{OPT}$.
\end{lemma}
\begin{proof}
By \cref{lem:offline-cost-singleton-3}, \cref{cla:eq-prerounding} and \cref{lem:opt-opt'}, we have $\E[S_2] \leq \frac{6}{k-1}\cdot \E[\mathrm{OPT}'] \leq \frac{6(2\beta+1)}{k-1} \cdot \mathrm{OPT}$.
\end{proof}

\begin{corollary}
\label{cor:streaming-cost-2.06}
    Let $\mathcal{C}_2$ denote the clustering returned by Line~\ref{line:cluster-B} of \cref{alg:insertion-only}, then $\E[\mathrm{cost}_G(\mathcal{C}_2)]=\E[P_2+S_2]\leq ( 2.06\beta+ \frac{6(2\beta+1)}{k-1})\cdot \mathrm{OPT}$.
\end{corollary}
\begin{proof}
\cref{cor:streaming-cost-2.06} follows from \cref{lem:eq-2.06}, \cref{lem:pivot-cost} and \cref{lem:singleton-cost}.
\end{proof}

\begin{proof}[Proof of \cref{thm:main-result-insertion}]
    \cref{thm:main-result-insertion} follows from \cref{cor:streaming-cost-3}, \cref{cor:streaming-cost-2.06} and \cref{lem:estimated-cost}.
\end{proof}
\textbf{Remark.} The reason our sampling-based approach works is mainly due to the fact that the rounding algorithm by Chawla et al.~\cite{CMSY15} is equivalent to the algorithm that first samples a subgraph $G'$ according to the prediction oracle and then runs the \textsc{Pivot} algorithm on $G'$. Therefore, if a Correlation Clustering algorithm $\mathcal{A}$ has a similar feature, i.e., can be viewed as a procedure that first obtains a core of the original graph (by using LP or other methods), and then applies the \textsc{Pivot} algorithm on the core, then we can get roughly the same approximation ratio as $\mathcal{A}$.

\section{Omitted details of \cref{sec:general}}
\label{sec:omitted-general}
\subsection{Proof of \cref{lem:general-pos-cost}}
    The positive cost of $\mathcal{C}_2$ on $H$ is
    \begin{align*}
    \cost_H^+(\mathcal{C}_2) = \cost_{H^+}(\mathcal{C}_2)
    &=\frac{1}{2}\sum_{C\in \mathcal{C}_2} \partial_{H^+}(C)\\
    &\le \frac{3}{2}\ln(n+1)\sum_{C\in \mathcal{C}_2} \vol_{H^+}(C)\\
    &\le\frac{3}{2}\ln(n+1)\left(\sum_{(u,v)\in E_H^+}w'_{uv}d_{uv}+ \sum_{C\in \mathcal{C}_2}\frac{V^*}{n}\right)\\
    &\le 3\ln(n+1)\cdot \sum_{(u,v)\in E_H^+}w'_{uv}d_{uv},
\end{align*}
where the second-to-last step follows from the triangle inequality, and the final step uses the fact that $V^*=\sum_{(u,v)\in E_H^+}w'_{uv}d_{uv}$ and $|\mathcal{C}_2|\le n$.

\subsection{Proof of \cref{lem:general-neg-cost}}
    Since
\begin{align*}
    \sum_{(u,v)\in E^-}(1-d_{uv}) \ge \sum_{C\in \mathcal{C}_2}\sum_{(u,v)\in E^-:u,v\in C}(1-d_{uv})
    &\ge \sum_{C\in \mathcal{C}_2}\sum_{(u,v)\in E^-:u,v\in C}\left(1-\frac{2}{3}\right)\\
    &= \frac{1}{3} \sum_{C\in \mathcal{C}_2}|(u,v)\in E^-:u,v\in C|,
\end{align*}
where the second step follows from $d_{uv} \le \frac{2}{3}$ by triangle inequality,
we have
\begin{align*}
    \cost_H^-(\mathcal{C}_2)=\sum_{C\in \mathcal{C}_2}|(u,v)\in E^-:u,v\in C| \le 3\sum_{(u,v)\in E^-}(1-d_{uv}).
\end{align*}

\subsection{Results under bounded degree graphs}
\label{subsec:d-regular}
If the input graph has bounded degree, then the adapted $\beta$-level predictor can be relaxed to the standard $\beta$-level predictor.
\begin{corollary}
\label{cor:general-reff}
    Let $\varepsilon\in (0,1/2)$ and $d,\beta \geq 1$.
    Given oracle access to a $\beta$-level predictor, there exists a single-pass streaming algorithm that, with high probability, achieves an $O(\log |E^-|+\beta d)$-approximation for Correlation Clustering on general graphs with maximum degree $d$ in dynamic streams. The algorithm uses $\tilde{O}(\varepsilon^{-2}n)$ words of space.
\end{corollary}
\begin{proof}
    Recall that the $\varepsilon$-spectral sparsifier $H^+$ is constructed using the algorithm of \cref{thm:dynamic-spectral}, which is based on effective resistance-based sampling~\cite{KMMMNST20fast}. 
    Specifically, each edge $(u,v)$ in the original graph $G^+$ is sampled with probability $p_{uv}\ge \frac{C\ln n}{\varepsilon^2}\cdot R_G(u,v)$, where $C$ is a sufficiently large constant. For any edge $(u,v)\in E_H^+$, its weight is assigned by $w'_{uv} = \frac{1}{p_{uv}} \le \frac{\varepsilon^2}{C\ln n \cdot R_G(u,v)}$.
    Since the maximum degree of $G$ is $d$, by \cref{lem:prop-eff}, we have $R_G(u,v) \ge \frac{1}{2}(\frac{1}{\deg(u)}+\frac{1}{\deg(v)})\ge \frac{1}{d}$ for all $u,v\in V$.
    It follows that $w'_{uv} \le \frac{\varepsilon^2 d}{C\ln n }$ for all $(u,v)\in E_H^+$.
    
    Then it suffices to follow the proof of \cref{lem:general-cost} to obtain
    \begin{align*}
        \cost_G(\mathcal{C}_2)
        &\le \frac{3\ln(n+1)}{1-\varepsilon}\cdot \sum_{(u,v)\in E_H^+}w'_{uv}d_{uv} + 3\sum_{(u,v)\in E^-}(1-d_{uv})\\
        &\le \frac{3\varepsilon^2 d\ln(n+1)}{(1-\varepsilon)C\ln n}\cdot \sum_{(u,v)\in E_H^+}d_{uv} + 3\sum_{(u,v)\in E^-}(1-d_{uv})\\
        &\le \max\left\{\frac{(1+\varepsilon)d}{C}\left(1+\frac{1}{n\ln n}\right),3\right\}\cdot \left(\sum_{(u,v)\in E^+}d_{uv}+\sum_{(u,v)\in E^-}(1-d_{uv})\right)\\
        &= O(\beta d)\cdot \opt,
    \end{align*}
    where the third step uses the fact that $\frac{3\varepsilon^2}{1-\varepsilon}\le 1+\varepsilon$ for any $\varepsilon \in (0,1/2)$, that $\frac{\ln(n+1)}{\ln n} \le 1+\frac{1}{n\ln n}$ for sufficiently large $n$, and that $E_H^+ \subseteq E^+$; and the last step follows from the definition of the $\beta$-level predictor (\cref{def:beta-level-predictor}).

    Combining \cref{lem:icml15} with the above analysis yields \cref{cor:general-reff}.
\end{proof}

\section{Additional experiments}
\label{sec:additional-exp}
In this section, we provide detailed descriptions of the datasets and predictors used in the experiments. Additionally, we present further experimental results.

\subsection{Detailed descriptions of datasets}
\label{subsec:detailed-dataset}
In this subsection, we give a detailed description of the real-world datasets used in our experiments. Recall that 
we use \textsc{EmailCore} \citep{LKF07, YBLG17}, \textsc{Facebook} \citep{ML12},  \textsc{LastFM}~\citep{Rozemberczki20characteristic}, and \textsc{DBLP}~\citep{Yang15defining} from the Stanford Large Network Dataset Collection \citep{snapnets}. 
We provide basic statistics about these datasets in \cref{tab:datasets}. 

\begin{table}[t]

\caption{Statistics of real-world datasets. Here, $n$ and $m$ denote the number of vertices and edges in the original datasets.}
\label{tab:datasets}
\begin{center}
\begin{tabular}{crr}
\toprule
\textbf{Datasets} & $n$ & $m$ \\ \midrule
\textsc{EmailCore} & \numprint{1005} & \numprint{25571} \\ \midrule
\textsc{Facebook} & \numprint{4039} & \numprint{88324} \\ \midrule
\textsc{LastFM} & \numprint{7624} & \numprint{27806} \\ \midrule
\textsc{DBLP} & \numprint{317080} & \numprint{1049866} \\ 
\bottomrule
\end{tabular}
\end{center}
\end{table}

\textsc{EmailCore} is a directed network with \numprint{1005} vertices and \numprint{25571} edges. This network is constructed based on email exchange data from a large European research institution. Each vertex represents a person in the institution. There is a directed edge $(u, v)$ in the network if person $u$ has sent at least one email to person $v$.

\textsc{Facebook} is an undirected network with \numprint{4039} vertices and \numprint{88324} edges. This network consists of friend lists of users from Facebook. Each vertex represents a user in Facebook. There is an undirected edge $(u, v)$ in the network if $u$ and $v$ are friends.
Due to the computational bottleneck of solving the LP, we only use its three ego-networks: \textsc{FB 0} ($n=\numprint{333}, m=\numprint{5038}$), \textsc{FB 414} ($n=\numprint{150}, m=\numprint{3386}$), \textsc{FB 3980} ($n=\numprint{52}, m=\numprint{292}$). 

\textsc{LastFM} is an undirected network with \numprint{7624} vertices and \numprint{27806} edges. This network is a social network of LastFM users,  collected from the public API. Each vertex represents a LastFM user from an Asian country. There is an undirected edge $(u, v)$ in the network if $u$ and $v$ are mutual followers.

\textsc{DBLP} is an undirected co-authorship network with \numprint{317080} vertices and \numprint{1049866} edges.
Each vertex represents an author. There is an undirected edge $(u, v)$ in the network if $u$ and $v$ publish at least one paper together. Ground-truth communities are defined based on publication venues: authors who have published in the same journal or conference belong to the same community. For our experiments, we use a sampled subgraph consisting of \numprint{10000} vertices.

\textbf{Remark.}
We treat the original edges in the datasets as positive edges and non-edges as negative implicitly. (For datasets used in experiments where binary classifiers are employed as predictors, the interpretation of positive and negative edges differs slightly. See \cref{subsec:detailed-predictor} for details.)
For directed networks, we convert all directed edges into undirected edges.
\emph{We highlight that since we consider labeled complete graphs in the experiments, the number of edges scales quadratically w.r.t. the number of vertices, which leads to non-trivial instance sizes.}

\subsection{Detailed descriptions of predictors}
\label{subsec:detailed-predictor}
\textbf{Noisy predictor.} We use this predictor for datasets with available optimal clusterings. We form this predictor by performing perturbations on optimal clusterings. Specifically, for any two vertices $u,v\in V$, if $u$ and $v$ are in different clusters in the optimal clustering, then we set the prediction $d_{uv}$ to be $1-\varepsilon_0$, otherwise $\varepsilon_0$, where $\varepsilon_0 \in (0,0.5)$. 
For synthetic datasets with $p>0.9$, we can assume that the ground truths are also optimal solutions. For other datasets, we use the powerful LP solver Gurobi \citep{gurobi} to get the optimal clusterings.

\textbf{Spectral embedding.} 
We use this predictor for \textsc{EmailCore} and \textsc{LastFM}. It first maps all the vertices to a $d$-dimensional Euclidean space using the graph Laplacian, then clusters all the vertices based on their embeddings. 
For any two vertices $u,v\in V$, we form the prediction $d_{uv}$ to be $1-\frac{\left\langle \vx_u,\vx_v\right\rangle}{\norm{\vx_u}\norm{\vx_v}}$, where $\vx_u, \vx_v \in \mathbb{R}^d$ are spectral embeddings of $u$ and $v$, and $\left\langle \vx_u,\vx_v\right\rangle$ is the dot product of $\vx_u$ and $\vx_v$. Note that a larger $d$ indicates a higher-quality predictor.

\textbf{Binary classifier.} 
We use this predictor for datasets with available ground-truth communities. 
This predictor is constructed by training a binary classifier (based on an MLP model) to predict whether two vertices belong to the same cluster using node features. 
In this setting, the goal of Correlation Clustering  aligns with that of community detection by treating edges between two vertices in the same (ground-truth) community as positive edges and edges between two vertices in different communities as negative edges.
The predictions provided by the binary classifier (i.e., binary values in $\{0, 1\}$) are then used as the pairwise distances $d_{uv}$ in our algorithms.

\subsection{Additional results}

\subsubsection{Performance of \cref{alg:insertion-only} on real-world datasets}
In this subsection, we present the results of our algorithm in insertion-only streams (\cref{alg:insertion-only}) on real-world datasets, as shown in \cref{fig:insertion-results-realworld}. 
The results show that under good prediction quality, \cref{alg:insertion-only} consistently outperforms other baselines across all datasets used. For example, in Figure~\ref{fig:insertion-results-realworld}(\subref{fig:insertion-fb0-vary-beta}), when $\beta \approx 1.2$, 
the average cost of \cref{alg:insertion-only} is $13\%$ lower than that of \textsc{CLMNPT21} and $17\%$ lower than that of \textsc{CKLPU24}. 
Besides, in Figure~\ref{fig:insertion-results-realworld}(\subref{fig:insertion-email-vary-d}), \cref{alg:insertion-only} reduces the clustering cost by up to $14\%$ compared to \textsc{CLMNPT21}. 
Even if the prediction quality is poor, \cref{alg:insertion-only} does not perform worse than \textsc{CM23} and achieves comparable performance to \textsc{CLMNPT21} (on \textsc{Facebook} subgraphs).

\begin{figure*}[t]
	\centering
        \begin{subfigure}{0.244\textwidth}
		\includegraphics[width=\textwidth]{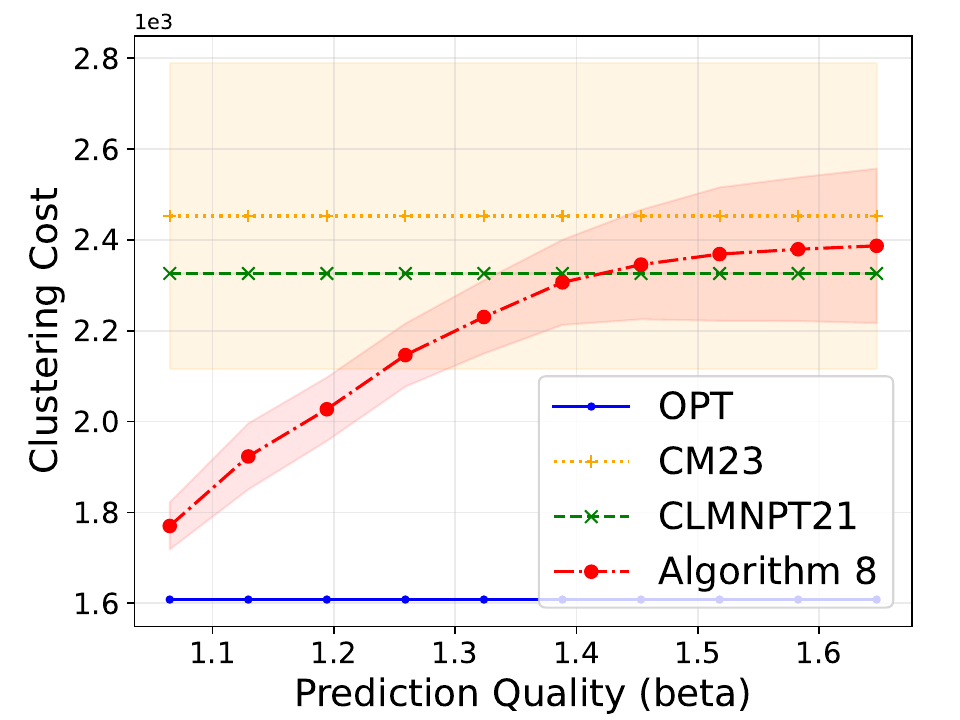}
		\caption{\textsc{FB 0}, vary $\beta$}
		\label{fig:insertion-fb0-vary-beta}
	\end{subfigure}
	\hfill
	\begin{subfigure}{0.244\textwidth}
		\includegraphics[width=\textwidth]{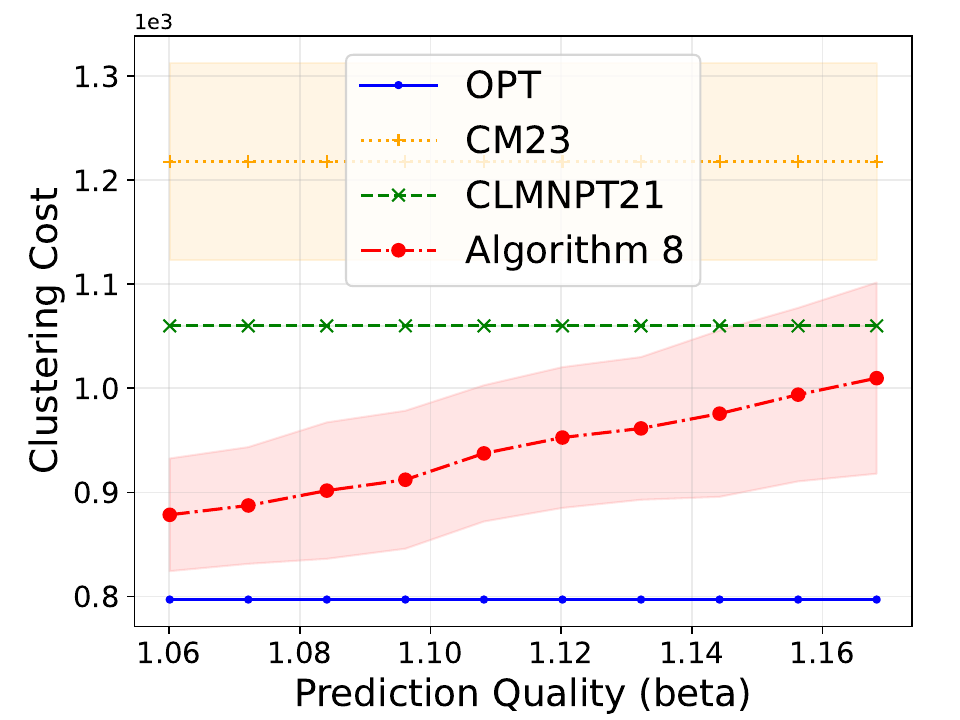}
		\caption{\textsc{FB 414}, vary $\beta$}
		\label{fig:insertion-fb414-vary-beta}
	\end{subfigure}
	\hfill
	\begin{subfigure}{0.244\textwidth}
		\includegraphics[width=\textwidth]{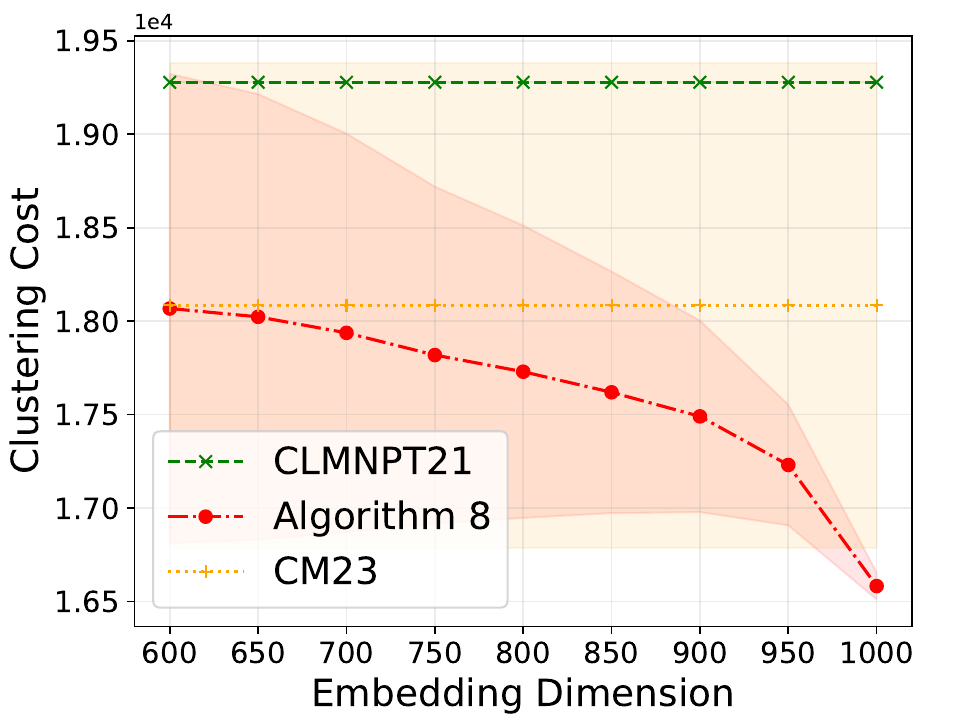}
		\caption{\textsc{EmailCore}, vary $d$}
		\label{fig:insertion-email-vary-d}
	\end{subfigure}
	\hfill
	\begin{subfigure}{0.244\textwidth}
		\includegraphics[width=\textwidth]{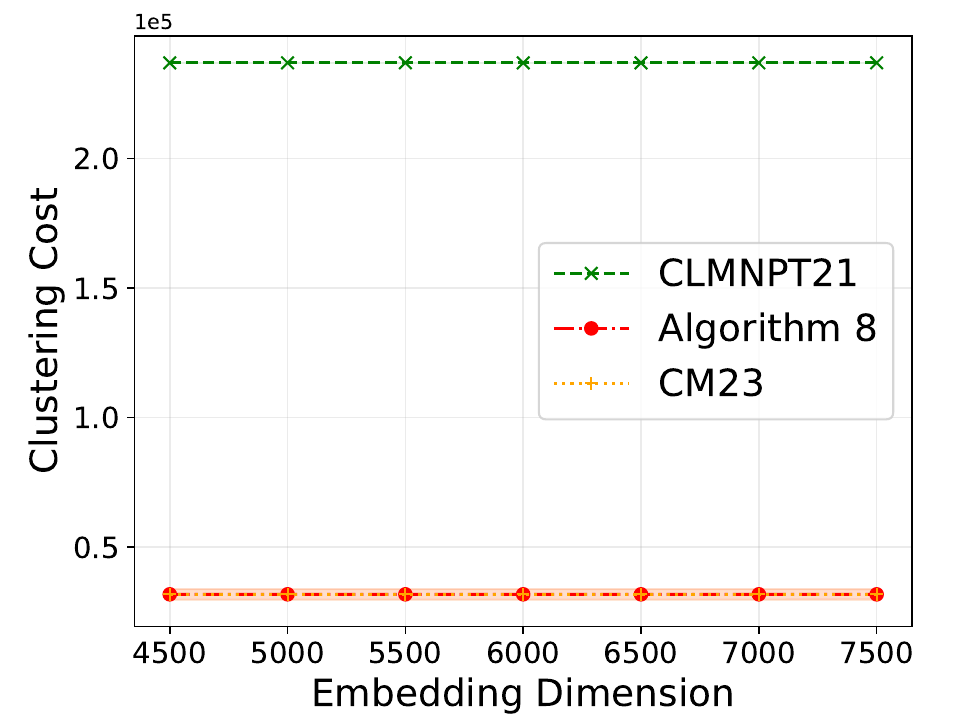}
		\caption{\textsc{LastFM}, vary $d$}
		\label{fig:insertion-lastfm-vary-d}
	\end{subfigure}
	\caption{Performance of \cref{alg:insertion-only} on real-world datasets. (\subref{fig:insertion-fb0-vary-beta})--(\subref{fig:insertion-fb414-vary-beta}) show the effect of prediction quality $\beta$ on two \textsc{Facebook} subgraphs, where we use noisy predictors. (\subref{fig:insertion-email-vary-d})--(\subref{fig:insertion-lastfm-vary-d}) examine the effect of the dimension $d$ of spectral embeddings on \textsc{EmailCore} and \textsc{LastFM}, where we use spectral embedding as the predictor. We set $k=25$ for (\subref{fig:insertion-fb0-vary-beta}), $k=15$ for (\subref{fig:insertion-fb414-vary-beta}), $k=10$ for (\subref{fig:insertion-email-vary-d}), and $k=50$ for (\subref{fig:insertion-lastfm-vary-d}).}
	\label{fig:insertion-results-realworld}
\end{figure*}

\begin{table}[t]

\caption{Clustering costs ($\times$ 1e3) of \cref{alg:insertion-only} with binary classifiers as predictors, compared to its non-learning counterpart. We set parameter $k=75$ across all datasets. The reported values are averaged over $5$ runs.}
\label{tab:binary-classifier-insertion}
\begin{center}
\begin{tabular}{crrrr}
\toprule
\diagbox{\textbf{Algorithm}}{\textbf{Dataset}} & \makecell{SBM\\($n=1200$)} & \makecell{SBM\\($n=2400$)} & \makecell{SBM\\($n=3600$)} & \textsc{DBLP}\\ \midrule
\textsc{CM23}    & \numprint{99.3}  & \numprint{385.7}   & \numprint{901.6}  & \numprint{236.4} \\
\textbf{\cref{alg:insertion-only}}   & \numprint{35.9} & \numprint{155.3} & \numprint{324.9}   & \numprint{214.9}  \\
\bottomrule
\end{tabular}
\end{center}
\end{table}

\subsubsection{Performance of \cref{alg:insertion-only} using binary classifiers as predictors}
In this subsection, we present the results of \cref{alg:insertion-only} using binary classifiers as predictors, as shown in \cref{tab:binary-classifier-insertion}. 
These experiments are performed on three SBM graphs with parameter $p=0.95$ and varying sizes, as well as the \textsc{DBLP} dataset (a sampled subgraph with \numprint{10000} vertices).
The results demonstrate that \cref{alg:insertion-only} consistently outperforms its non-learning counterpart across all datasets. For instance, on the SBM graphs with \numprint{1200} and \numprint{3600} vertices, \cref{alg:insertion-only} achieves a $64\%$ reduction in clustering cost compared to \textsc{CM23}.

\subsubsection{Running time of our algorithms}
In this subsection, we present the running time of our algorithms for complete graphs on three \textsc{Facebook} subgraphs, compared to their non-learning counterparts, as shown in \cref{tab:running-time-dynamic} (\cref{alg:dynamic-stream}) and \cref{tab:running-time-insertion} (\cref{alg:insertion-only}). The results show that our learning-augmented algorithms do not introduce significant time overheads. The slight increase in running time is due to the additional steps of querying the oracles and calculating the costs of two clusterings. These steps are both reasonable and acceptable. Moreover, in the streaming setting, space efficiency is typically the primary focus.

\begin{table}[t]

\caption{Running time (ms) of \cref{alg:dynamic-stream} (for dynamic streams) on \textsc{Facebook} subgraphs, compared to its non-learning counterpart. For \textsc{FB 0}, we set $\beta = 1.19$. For \textsc{FB 414}, we set $\beta = 1.12$. For \textsc{FB 3980}, we set $\beta = 1.19$.}
\label{tab:running-time-dynamic}
\begin{center}
\begin{tabular}{crrr}
\toprule
\diagbox{\textbf{Algorithm}}{\textbf{Dataset}}        & \textsc{FB 0} & \textsc{FB 414} & \textsc{FB 3980} \\ \midrule
\textsc{CKLPU24}  & \numprint{1738.16} & \numprint{165.55} & \numprint{7.32}        \\
\textbf{\cref{alg:dynamic-stream}} & \numprint{1639.22} & \numprint{163.35} & \numprint{7.69}        \\ 
\bottomrule
\end{tabular}
\end{center}
\end{table}

\begin{table}[t]

\caption{Running time (ms) of \cref{alg:insertion-only} (for insertion-only streams) on \textsc{Facebook} subgraphs, compared to its non-learning counterpart. For \textsc{FB 0}, we set $\beta = 1.19$. For \textsc{FB 414}, we set $\beta = 1.12$. For \textsc{FB 3980}, we set $\beta = 1.19$.}
\label{tab:running-time-insertion}
\begin{center}
\begin{tabular}{crrr}
\toprule
\diagbox{\textbf{Algorithm}}{\textbf{Dataset}}       &  \textsc{FB 0} & \textsc{FB 414} & \textsc{FB 3980} \\ \midrule
\textsc{CM23} & \numprint{30.65} & \numprint{6.67} & \numprint{0.97}         \\
\textbf{\cref{alg:insertion-only}}   & \numprint{81.31} & \numprint{16.58} & \numprint{2.12}         \\ 
\bottomrule
\end{tabular}
\end{center}
\end{table}

\end{document}

%% file: math_commands.tex
\usepackage{amsmath,amsfonts,bm}

\def\eqref#1{equation~\ref{#1}}
\def\1{\bm{1}}

\providecommand{\norm}[1]{\lVert#1\rVert}
\def\vzero{{\bm{0}}}

\def\vx{{\bm{x}}}

\DeclareMathAlphabet{\mathsfit}{\encodingdefault}{\sfdefault}{m}{sl}
\SetMathAlphabet{\mathsfit}{bold}{\encodingdefault}{\sfdefault}{bx}{n}

\def\sR{{\mathbb{R}}}

\newcommand{\E}{\mathbb{E}}

%% file: main.bbl
\begin{thebibliography}{10}

\bibitem{ACGSW25}
Anders Aamand, Justin~Y. Chen, Siddharth Gollapudi, Sandeep Silwal, and Hao Wu.
\newblock Learning-augmented frequent directions.
\newblock In {\em the Thirteenth International Conference on Learning Representations ({ICLR})}, 2025.

\bibitem{ACNSV23}
Anders Aamand, Justin~Y. Chen, Huy~L{\^{e}} Nguyen, Sandeep Silwal, and Ali Vakilian.
\newblock Improved frequency estimation algorithms with and without predictions.
\newblock In {\em Advances in Neural Information Processing Systems 36: Annual Conference on Neural Information Processing Systems (NeurIPS)}, 2023.

\bibitem{AgrawalBGOT22}
Priyank Agrawal, Eric Balkanski, Vasilis Gkatzelis, Tingting Ou, and Xizhi Tan.
\newblock {Learning-Augmented Mechanism Design: Leveraging Predictions for Facility Location}.
\newblock In {\em 23rd {ACM} Conference on Economics and Computation ({EC})}, pages 497--528, 2022.

\bibitem{ACGMW21}
Kook~Jin Ahn, Graham Cormode, Sudipto Guha, Andrew McGregor, and Anthony Wirth.
\newblock Correlation clustering in data streams.
\newblock {\em Algorithmica}, 83(7):1980--2017, 2021.
\newblock Conference version in Proceedings of the 32nd International Conference on Machine Learning, {ICML} 2015.

\bibitem{ACN08}
Nir Ailon, Moses Charikar, and Alantha Newman.
\newblock Aggregating inconsistent information: Ranking and clustering.
\newblock {\em J. {ACM}}, 55(5):23:1--23:27, 2008.

\bibitem{Angelopoulos24online}
Spyros Angelopoulos, Christoph D{\"{u}}rr, Shendan Jin, Shahin Kamali, and Marc~P. Renault.
\newblock Online computation with untrusted advice.
\newblock {\em J. Comput. Syst. Sci.}, 144:103545, 2024.

\bibitem{Antoniadis23online}
Antonios Antoniadis, Christian Coester, Marek Eli{\'{a}}s, Adam Polak, and Bertrand Simon.
\newblock Online metric algorithms with untrusted predictions.
\newblock {\em {ACM} Trans. Algorithms}, 19(2):19:1--19:34, 2023.

\bibitem{AEPV25}
Antonios Antoniadis, Marek Eli{\'{a}}s, Adam Polak, and Moritz Venzin.
\newblock Approximation algorithms for combinatorial optimization with predictions.
\newblock In {\em the Thirteenth International Conference on Learning Representations ({ICLR})}, 2025.

\bibitem{Antoniadis23secretary}
Antonios Antoniadis, Themis Gouleakis, Pieter Kleer, and Pavel Kolev.
\newblock Secretary and online matching problems with machine learned advice.
\newblock {\em Discret. Optim.}, 48(Part 2):100778, 2023.

\bibitem{AKP25}
Sepehr Assadi, Sanjeev Khanna, and Aaron Putterman.
\newblock Correlation clustering and (de)sparsification: Graph sketches can match classical algorithms.
\newblock In {\em Proceedings of the 57th Annual {ACM} Symposium on Theory of Computing ({STOC})}, 2025.

\bibitem{ASW23}
Sepehr Assadi, Vihan Shah, and Chen Wang.
\newblock Streaming algorithms and lower bounds for estimating correlation clustering cost.
\newblock In {\em Advances in Neural Information Processing Systems 36: Annual Conference on Neural Information Processing Systems (NeurIPS)}, 2023.

\bibitem{AW22}
Sepehr Assadi and Chen Wang.
\newblock Sublinear time and space algorithms for correlation clustering via sparse-dense decompositions.
\newblock In {\em 13th Innovations in Theoretical Computer Science Conference ({ITCS})}, volume 215 of {\em LIPIcs}, pages 10:1--10:20, 2022.

\bibitem{Bamas20primal}
{\'{E}}tienne Bamas, Andreas Maggiori, and Ola Svensson.
\newblock The primal-dual method for learning augmented algorithms.
\newblock In {\em Advances in Neural Information Processing Systems 33: Annual Conference on Neural Information Processing Systems (NeurIPS)}, 2020.

\bibitem{Banerjee23graph}
Siddhartha Banerjee, Vincent Cohen{-}Addad, Anupam Gupta, and Zhouzi Li.
\newblock Graph searching with predictions.
\newblock In {\em 14th Innovations in Theoretical Computer Science Conference ({ITCS})}, volume 251 of {\em LIPIcs}, pages 12:1--12:24, 2023.

\bibitem{BBC04}
Nikhil Bansal, Avrim Blum, and Shuchi Chawla.
\newblock Correlation clustering.
\newblock {\em Machine learning}, 56:89--113, 2004.

\bibitem{BCCGM24}
Soheil Behnezhad, Moses Charikar, Vincent Cohen{-}Addad, Alma Ghafari, and Weiyun Ma.
\newblock Correlation clustering beyond the pivot algorithm.
\newblock In {\em Forty-second International Conference on Machine Learning ({ICML})}, 2025.

\bibitem{BCMT22}
Soheil Behnezhad, Moses Charikar, Weiyun Ma, and Li{-}Yang Tan.
\newblock Almost 3-approximate correlation clustering in constant rounds.
\newblock In {\em 63rd {IEEE} Annual Symposium on Foundations of Computer Science ({FOCS})}, pages 720--731, 2022.

\bibitem{BCMT23}
Soheil Behnezhad, Moses Charikar, Weiyun Ma, and Li{-}Yang Tan.
\newblock Single-pass streaming algorithms for correlation clustering.
\newblock In {\em Proceedings of the 2023 {ACM-SIAM} Symposium on Discrete Algorithms ({SODA})}, pages 819--849, 2023.

\bibitem{BK96}
Andr{\'{a}}s~A. Bencz{\'{u}}r and David~R. Karger.
\newblock Approximating $s$-$t$ minimum cuts in $\tilde{O}(n^2)$ time.
\newblock In {\em Proceedings of the Twenty-Eighth Annual {ACM} Symposium on the Theory of Computing ({STOC})}, pages 47--55, 1996.

\bibitem{Brand24dynamic}
Jan van~den Brand, Sebastian Forster, Yasamin Nazari, and Adam Polak.
\newblock On dynamic graph algorithms with predictions.
\newblock In {\em Proceedings of the 2024 {ACM-SIAM} Symposium on Discrete Algorithms ({SODA})}, pages 3534--3557, 2024.

\bibitem{Braverman24learning}
Vladimir Braverman, Prathamesh Dharangutte, Vihan Shah, and Chen Wang.
\newblock Learning-augmented maximum independent set.
\newblock In {\em Approximation, Randomization, and Combinatorial Optimization. Algorithms and Techniques ({APPROX/RANDOM})}, volume 317 of {\em LIPIcs}, pages 24:1--24:18, 2024.

\bibitem{CKLPU24}
M{\'{e}}lanie Cambus, Fabian Kuhn, Etna Lindy, Shreyas Pai, and Jara Uitto.
\newblock A $(3+\varepsilon)$-approximate correlation clustering algorithm in dynamic streams.
\newblock In {\em Proceedings of the 2024 {ACM-SIAM} Symposium on Discrete Algorithms ({SODA})}, pages 2861--2880, 2024.

\bibitem{CCL+25}
Nairen Cao, Vincent Cohen{-}Addad, Euiwoong Lee, Shi Li, David~Rasmussen Lolck, Alantha Newman, Mikkel Thorup, Lukas Vogl, Shuyi Yan, and Hanwen Zhang.
\newblock Solving the correlation cluster lp in sublinear time.
\newblock In {\em Proceedings of the 57th Annual {ACM} Symposium on Theory of Computing ({STOC})}, 2025.

\bibitem{CCLLNV24}
Nairen Cao, Vincent Cohen{-}Addad, Euiwoong Lee, Shi Li, Alantha Newman, and Lukas Vogl.
\newblock Understanding the cluster linear program for correlation clustering.
\newblock In {\em Proceedings of the 56th Annual {ACM} Symposium on Theory of Computing ({STOC})}, pages 1605--1616, 2024.

\bibitem{CHS24}
Nairen Cao, Shang{-}En Huang, and Hsin{-}Hao Su.
\newblock Breaking 3-factor approximation for correlation clustering in polylogarithmic rounds.
\newblock In {\em Proceedings of the 2024 {ACM-SIAM} Symposium on Discrete Algorithms ({SODA})}, pages 4124--4154, 2024.

\bibitem{CLY25}
Nairen Cao, Shi Li, and Jia Ye.
\newblock Simultaneously approximating all norms for massively parallel correlation clustering.
\newblock In {\em 52rd International Colloquium on Automata, Languages, and Programming ({ICALP})}, LIPIcs, 2025.

\bibitem{Caragiannis024}
Ioannis Caragiannis and Georgios Kalantzis.
\newblock {Randomized Learning-Augmented Auctions with Revenue Guarantees}.
\newblock In {\em Proceedings of the Thirty-Third International Joint Conference on Artificial Intelligence ({IJCAI})}, pages 2687--2694, 2024.

\bibitem{CKP08}
Deepayan Chakrabarti, Ravi Kumar, and Kunal Punera.
\newblock A graph-theoretic approach to webpage segmentation.
\newblock In {\em Proceedings of the 17th International Conference on World Wide Web ({WWW})}, pages 377--386, 2008.

\bibitem{CM23}
Sayak Chakrabarty and Konstantin Makarychev.
\newblock Single-pass pivot algorithm for correlation clustering. keep it simple!
\newblock In {\em Advances in Neural Information Processing Systems 36: Annual Conference on Neural Information Processing Systems (NeurIPS)}, 2023.

\bibitem{CGS17}
Moses Charikar, Neha Gupta, and Roy Schwartz.
\newblock Local guarantees in graph cuts and clustering.
\newblock In {\em Integer Programming and Combinatorial Optimization - 19th International Conference ({IPCO})}, volume 10328 of {\em Lecture Notes in Computer Science}, pages 136--147, 2017.

\bibitem{CGW05}
Moses Charikar, Venkatesan Guruswami, and Anthony Wirth.
\newblock Clustering with qualitative information.
\newblock {\em J. Comput. Syst. Sci.}, 71(3):360--383, 2005.

\bibitem{CMSY15}
Shuchi Chawla, Konstantin Makarychev, Tselil Schramm, and Grigory Yaroslavtsev.
\newblock Near optimal {LP} rounding algorithm for correlation clustering on complete and complete {$k$}-partite graphs.
\newblock In {\em Proceedings of the Forty-Seventh Annual {ACM} on Symposium on Theory of Computing ({STOC})}, pages 219--228, 2015.

\bibitem{CEILNRSWWZ22}
Justin~Y. Chen, Talya Eden, Piotr Indyk, Honghao Lin, Shyam Narayanan, Ronitt Rubinfeld, Sandeep Silwal, Tal Wagner, David~P. Woodruff, and Michael Zhang.
\newblock Triangle and four cycle counting with predictions in graph streams.
\newblock In {\em 10th International Conference on Learning Representations ({ICLR})}, 2022.

\bibitem{faster-graph}
Justin~Y. Chen, Sandeep Silwal, Ali Vakilian, and Fred Zhang.
\newblock Faster fundamental graph algorithms via learned predictions.
\newblock In {\em International Conference on Machine Learning ({ICML})}, volume 162 of {\em Proceedings of Machine Learning Research}, pages 3583--3602, 2022.

\bibitem{CDGLP24}
Vincent Cohen{-}Addad, Tommaso d'Orsi, Anupam Gupta, Euiwoong Lee, and Debmalya Panigrahi.
\newblock Learning-augmented approximation algorithms for maximum cut and related problems.
\newblock In {\em Advances in Neural Information Processing Systems 37: Annual Conference on Neural Information Processing Systems ({NeurIPS})}, 2024.

\bibitem{CLMNP21}
Vincent Cohen{-}Addad, Silvio Lattanzi, Slobodan Mitrovic, Ashkan Norouzi{-}Fard, Nikos Parotsidis, and Jakub Tarnawski.
\newblock Correlation clustering in constant many parallel rounds.
\newblock In {\em International Conference on Machine Learning ({ICML})}, volume 139 of {\em Proceedings of Machine Learning Research}, pages 2069--2078, 2021.

\bibitem{CLLN23}
Vincent Cohen{-}Addad, Euiwoong Lee, Shi Li, and Alantha Newman.
\newblock Handling correlated rounding error via preclustering: {A} 1.73-approximation for correlation clustering.
\newblock In {\em 64th {IEEE} Annual Symposium on Foundations of Computer Science ({FOCS})}, pages 1082--1104, 2023.

\bibitem{CLN22}
Vincent Cohen{-}Addad, Euiwoong Lee, and Alantha Newman.
\newblock Correlation clustering with sherali-adams.
\newblock In {\em 63rd {IEEE} Annual Symposium on Foundations of Computer Science ({FOCS})}, pages 651--661, 2022.

\bibitem{CLPTYZ24}
Vincent Cohen{-}Addad, David~Rasmussen Lolck, Marcin Pilipczuk, Mikkel Thorup, Shuyi Yan, and Hanwen Zhang.
\newblock Combinatorial correlation clustering.
\newblock In {\em Proceedings of the 56th Annual {ACM} Symposium on Theory of Computing ({STOC})}, pages 1617--1628, 2024.

\bibitem{Colini-Baldeschi24}
Riccardo Colini{-}Baldeschi, Sophie Klumper, Guido Sch{\"{a}}fer, and Artem Tsikiridis.
\newblock To trust or not to trust: Assignment mechanisms with predictions in the private graph model.
\newblock In {\em Proceedings of the 25th {ACM} Conference on Economics and Computation ({EC})}, pages 1134--1154, 2024.

\bibitem{DMM24}
Mina Dalirrooyfard, Konstantin Makarychev, and Slobodan Mitrovic.
\newblock Pruned pivot: Correlation clustering algorithm for dynamic, parallel, and local computation models.
\newblock In {\em Forty-first International Conference on Machine Learning ({ICML})}, 2024.

\bibitem{DMN24}
Sami Davies, Benjamin Moseley, and Heather Newman.
\newblock Simultaneously approximating all $\ell_p$-norms in correlation clustering.
\newblock In {\em 51st International Colloquium on Automata, Languages, and Programming ({ICALP})}, volume 297 of {\em LIPIcs}, pages 52:1--52:20, 2024.

\bibitem{Davies23predictive}
Sami Davies, Benjamin Moseley, Sergei Vassilvitskii, and Yuyan Wang.
\newblock Predictive flows for faster ford-fulkerson.
\newblock In {\em International Conference on Machine Learning ({ICML})}, volume 202 of {\em Proceedings of Machine Learning Research}, pages 7231--7248, 2023.

\bibitem{DEFI06}
Erik~D. Demaine, Dotan Emanuel, Amos Fiat, and Nicole Immorlica.
\newblock Correlation clustering in general weighted graphs.
\newblock {\em Theor. Comput. Sci.}, 361(2-3):172--187, 2006.

\bibitem{DePavia24learning}
Adela~Frances DePavia, Erasmo Tani, and Ali Vakilian.
\newblock Learning-based algorithms for graph searching problems.
\newblock In {\em International Conference on Artificial Intelligence and Statistics (AISTATS)}, volume 238 of {\em Proceedings of Machine Learning Research}, pages 928--936, 2024.

\bibitem{Dinitz21faster}
Michael Dinitz, Sungjin Im, Thomas Lavastida, Benjamin Moseley, and Sergei Vassilvitskii.
\newblock Faster matchings via learned duals.
\newblock In {\em Advances in Neural Information Processing Systems 34: Annual Conference on Neural Information Processing Systems (NeurIPS)}, pages 10393--10406, 2021.

\bibitem{Dong25maxcut}
Yinhao Dong, Pan Peng, and Ali Vakilian.
\newblock Learning-augmented streaming algorithms for approximating {MAX-CUT}.
\newblock In {\em 16th Innovations in Theoretical Computer Science Conference ({ITCS})}, volume 325 of {\em LIPIcs}, pages 44:1--44:24, 2025.

\bibitem{support}
Talya Eden, Piotr Indyk, Shyam Narayanan, Ronitt Rubinfeld, Sandeep Silwal, and Tal Wagner.
\newblock Learning-based support estimation in sublinear time.
\newblock In {\em 9th International Conference on Learning Representations ({ICLR})}, 2021.

\bibitem{ICLR22}
Jon~C. Ergun, Zhili Feng, Sandeep Silwal, David~P. Woodruff, and Samson Zhou.
\newblock Learning-augmented {$k$}-means clustering.
\newblock In {\em 10th International Conference on Learning Representations ({ICLR})}, 2022.

\bibitem{Feigenbaum05semi}
Joan Feigenbaum, Sampath Kannan, Andrew McGregor, Siddharth Suri, and Jian Zhang.
\newblock On graph problems in a semi-streaming model.
\newblock {\em Theor. Comput. Sci.}, 348(2-3):207--216, 2005.

\bibitem{FV20}
Paolo Ferragina and Giorgio Vinciguerra.
\newblock The pgm-index: A fully-dynamic compressed learned index with provable worst-case bounds.
\newblock {\em Proc. {VLDB} Endow.}, 13(8):1162--1175, 2020.

\bibitem{GVY96}
Naveen Garg, Vijay~V. Vazirani, and Mihalis Yannakakis.
\newblock Approximate max-flow min-(multi)cut theorems and their applications.
\newblock {\em {SIAM} J. Comput.}, 25(2):235--251, 1996.

\bibitem{Ghoshal24constraint}
Suprovat Ghoshal, Konstantin Makarychev, and Yury Markarychev.
\newblock Constraint satisfaction problems with advice.
\newblock In {\em Proceedings of the 2025 Annual {ACM-SIAM} Symposium on Discrete Algorithms ({SODA})}, pages 1202--1221, 2025.

\bibitem{GkatzelisKST22}
Vasilis Gkatzelis, Kostas Kollias, Alkmini Sgouritsa, and Xizhi Tan.
\newblock {Improved Price of Anarchy via Predictions}.
\newblock In {\em 23rd {ACM} Conference on Economics and Computation ({EC})}, pages 529--557, 2022.

\bibitem{gurobi}
{Gurobi Optimization, LLC}.
\newblock {Gurobi Optimizer Reference Manual}, 2023.

\bibitem{HIA24}
Holger S.~G. Heidrich, Jannik Irmai, and Bjoern Andres.
\newblock A 4-approximation algorithm for min max correlation clustering.
\newblock In {\em International Conference on Artificial Intelligence and Statistics ({AISTATS})}, volume 238 of {\em Proceedings of Machine Learning Research}, pages 1945--1953, 2024.

\bibitem{Henzinger24complexity}
Monika Henzinger, Barna Saha, Martin~P. Seybold, and Christopher Ye.
\newblock On the complexity of algorithms with predictions for dynamic graph problems.
\newblock In {\em 15th Innovations in Theoretical Computer Science Conference ({ITCS})}, volume 287 of {\em LIPIcs}, pages 62:1--62:25, 2024.

\bibitem{frequency}
Chen{-}Yu Hsu, Piotr Indyk, Dina Katabi, and Ali Vakilian.
\newblock Learning-based frequency estimation algorithms.
\newblock In {\em 7th International Conference on Learning Representations ({ICLR})}, 2019.

\bibitem{Im21online}
Sungjin Im, Ravi Kumar, Mahshid~Montazer Qaem, and Manish Purohit.
\newblock Online knapsack with frequency predictions.
\newblock In {\em Advances in Neural Information Processing Systems 34: Annual Conference on Neural Information Processing Systems (NeurIPS)}, pages 2733--2743, 2021.

\bibitem{Indyk19learning}
Piotr Indyk, Ali Vakilian, and Yang Yuan.
\newblock Learning-based low-rank approximations.
\newblock In {\em Advances in Neural Information Processing Systems 32: Annual Conference on Neural Information Processing Systems (NeurIPS)}, pages 7400--7410, 2019.

\bibitem{JLLRW20}
Tanqiu Jiang, Yi~Li, Honghao Lin, Yisong Ruan, and David~P. Woodruff.
\newblock Learning-augmented data stream algorithms.
\newblock In {\em 8th International Conference on Learning Representations ({ICLR})}, 2020.

\bibitem{JST11}
Hossein Jowhari, Mert Saglam, and G{\'{a}}bor Tardos.
\newblock Tight bounds for {$L_p$} samplers, finding duplicates in streams, and related problems.
\newblock In {\em Proceedings of the 30th {ACM} {SIGMOD-SIGACT-SIGART} Symposium on Principles of Database Systems ({PODS})}, pages 49--58, 2011.

\bibitem{KMZ19}
Sanchit Kalhan, Konstantin Makarychev, and Timothy Zhou.
\newblock Correlation clustering with local objectives.
\newblock In {\em Advances in Neural Information Processing Systems 32: Annual Conference on Neural Information Processing Systems ({NeurIPS})}, pages 9341--9350, 2019.

\bibitem{KMMMNST20fast}
Michael Kapralov, Aida Mousavifar, Cameron Musco, Christopher Musco, Navid Nouri, Aaron Sidford, and Jakab Tardos.
\newblock Fast and space efficient spectral sparsification in dynamic streams.
\newblock In {\em Proceedings of the 2020 {ACM-SIAM} Symposium on Discrete Algorithms ({SODA})}, pages 1814--1833, 2020.

\bibitem{KL11}
Jonathan~A. Kelner and Alex Levin.
\newblock Spectral sparsification in the semi-streaming setting.
\newblock In {\em 28th International Symposium on Theoretical Aspects of Computer Science ({STACS})}, volume~9 of {\em LIPIcs}, pages 440--451, 2011.

\bibitem{KYNK14}
Sungwoong Kim, Chang~Dong Yoo, Sebastian Nowozin, and Pushmeet Kohli.
\newblock Image segmentation using higher-order correlation clustering.
\newblock {\em {IEEE} Trans. Pattern Anal. Mach. Intell.}, 36(9):1761--1774, 2014.

\bibitem{kuroki2024query}
Yuko Kuroki, Atsushi Miyauchi, Francesco Bonchi, and Wei Chen.
\newblock Query-efficient correlation clustering with noisy oracle.
\newblock In {\em Advances in Neural Information Processing Systems 37: Annual Conference on Neural Information Processing Systems ({NeurIPS})}, 2024.

\bibitem{Scheduling}
Silvio Lattanzi, Thomas Lavastida, Benjamin Moseley, and Sergei Vassilvitskii.
\newblock Online scheduling via learned weights.
\newblock In {\em Proceedings of the 2020 {ACM-SIAM} Symposium on Discrete Algorithms ({SODA})}, pages 1859--1877, 2020.

\bibitem{LSV23}
Silvio Lattanzi, Ola Svensson, and Sergei Vassilvitskii.
\newblock Speeding up bellman ford via minimum violation permutations.
\newblock In {\em International Conference on Machine Learning ({ICML})}, volume 202 of {\em Proceedings of Machine Learning Research}, pages 18584--18598, 2023.

\bibitem{LKF07}
Jure Leskovec, Jon~M. Kleinberg, and Christos Faloutsos.
\newblock Graph evolution: Densification and shrinking diameters.
\newblock {\em {ACM} Trans. Knowl. Discov. Data}, 1(1):2, 2007.

\bibitem{snapnets}
Jure Leskovec and Andrej Krevl.
\newblock {SNAP Datasets}: {Stanford} large network dataset collection.
\newblock \url{http://snap.stanford.edu/data}, June 2014.

\bibitem{Li23learning}
Yi~Li, Honghao Lin, Simin Liu, Ali Vakilian, and David~P. Woodruff.
\newblock Learning the positions in countsketch.
\newblock In {\em 11th International Conference on Learning Representations ({ICLR})}, 2023.

\bibitem{LLW22}
Honghao Lin, Tian Luo, and David~P. Woodruff.
\newblock Learning augmented binary search trees.
\newblock In {\em International Conference on Machine Learning ({ICML})}, volume 162 of {\em Proceedings of Machine Learning Research}, pages 13431--13440, 2022.

\bibitem{liu23predicted}
Quanquan~C. Liu and Vaidehi Srinivas.
\newblock The predicted-deletion dynamic model: Taking advantage of {ML} predictions, for free.
\newblock {\em CoRR}, abs/2307.08890, 2023.

\bibitem{lovasz1993random}
L{\'a}szl{\'o} Lov{\'a}sz.
\newblock Random walks on graphs: A survey.
\newblock {\em Combinatorics, Paul Erd{\H{o}}s is Eighty}, 2(1):1--46, 1993.

\bibitem{LuW024}
Pinyan Lu, Zongqi Wan, and Jialin Zhang.
\newblock {Competitive Auctions with Imperfect Predictions}.
\newblock In {\em Proceedings of the 25th {ACM} Conference on Economics and Computation ({EC})}, pages 1155--1183, 2024.

\bibitem{caching}
Thodoris Lykouris and Sergei Vassilvitskii.
\newblock Competitive caching with machine learned advice.
\newblock {\em J. {ACM}}, 68(4):24:1--24:25, 2021.

\bibitem{ML12}
Julian~J. McAuley and Jure Leskovec.
\newblock Learning to discover social circles in ego networks.
\newblock In {\em Advances in Neural Information Processing Systems 25: 26th Annual Conference on Neural Information Processing Systems (NIPS)}, pages 548--556, 2012.

\bibitem{Mitzenmacher18model}
Michael Mitzenmacher.
\newblock A model for learned bloom filters and optimizing by sandwiching.
\newblock In {\em Advances in Neural Information Processing Systems 31: Annual Conference on Neural Information Processing Systems (NeurIPS)}, pages 462--471, 2018.

\bibitem{MV21}
Michael Mitzenmacher and Sergei Vassilvitskii.
\newblock {\em Algorithms with Predictions}, page 646–662.
\newblock Cambridge University Press, 2021.

\bibitem{ICLR23}
Thy Nguyen, Anamay Chaturvedi, and Huy~Le Nguyen.
\newblock Improved learning-augmented algorithms for {$k$}-means and {$k$}-medians clustering.
\newblock In {\em 11th International Conference on Learning Representations ({ICLR})}, 2023.

\bibitem{PM18}
Gregory~J. Puleo and Olgica Milenkovic.
\newblock Correlation clustering and biclustering with locally bounded errors.
\newblock {\em {IEEE} Trans. Inf. Theory}, 64(6):4105--4119, 2018.

\bibitem{Purohit18improving}
Manish Purohit, Zoya Svitkina, and Ravi Kumar.
\newblock Improving online algorithms via {ML} predictions.
\newblock In {\em Advances in Neural Information Processing Systems 31: Annual Conference on Neural Information Processing Systems (NeurIPS)}, pages 9684--9693, 2018.

\bibitem{Rozemberczki20characteristic}
Benedek Rozemberczki and Rik Sarkar.
\newblock Characteristic functions on graphs: Birds of a feather, from statistical descriptors to parametric models.
\newblock In {\em 29th {ACM} International Conference on Information and Knowledge Management (CIKM)}, pages 1325--1334, 2020.

\bibitem{Sato23fast}
Atsuki Sato and Yusuke Matsui.
\newblock Fast partitioned learned bloom filter.
\newblock In {\em Advances in Neural Information Processing Systems 36: Annual Conference on Neural Information Processing Systems (NeurIPS)}, 2023.

\bibitem{SDELM21}
Jessica Shi, Laxman Dhulipala, David Eisenstat, Jakub Lacki, and Vahab~S. Mirrokni.
\newblock Scalable community detection via parallel correlation clustering.
\newblock {\em Proc. {VLDB} Endow.}, 14(11):2305--2313, 2021.

\bibitem{Silwal23kwikbucks}
Sandeep Silwal, Sara Ahmadian, Andrew Nystrom, Andrew McCallum, Deepak Ramachandran, and Seyed~Mehran Kazemi.
\newblock Kwikbucks: Correlation clustering with cheap-weak and expensive-strong signals.
\newblock In {\em 11th International Conference on Learning Representations ({ICLR})}, 2023.

\bibitem{ST11}
Daniel~A. Spielman and Shang{-}Hua Teng.
\newblock Spectral sparsification of graphs.
\newblock {\em {SIAM} J. Comput.}, 40(4):981--1025, 2011.

\bibitem{Swamy04}
Chaitanya Swamy.
\newblock Correlation clustering: Maximizing agreements via semidefinite programming.
\newblock In {\em Proceedings of the Fifteenth Annual {ACM-SIAM} Symposium on Discrete Algorithms ({SODA})}, pages 526--527, 2004.

\bibitem{Vaidya21partitioned}
Kapil Vaidya, Eric Knorr, Michael Mitzenmacher, and Tim Kraska.
\newblock Partitioned learned bloom filters.
\newblock In {\em 9th International Conference on Learning Representations ({ICLR})}, 2021.

\bibitem{Veldt18community}
Nate Veldt, David~F. Gleich, and Anthony Wirth.
\newblock A correlation clustering framework for community detection.
\newblock In {\em Proceedings of the 2018 World Wide Web Conference on World Wide Web ({WWW})}, pages 439--448, 2018.

\bibitem{XuL22}
Chenyang Xu and Pinyan Lu.
\newblock {Mechanism Design with Predictions}.
\newblock In {\em Proceedings of the Thirty-First International Joint Conference on Artificial Intelligence ({IJCAI})}, pages 571--577, 2022.

\bibitem{Yang15defining}
Jaewon Yang and Jure Leskovec.
\newblock Defining and evaluating network communities based on ground-truth.
\newblock {\em Knowl. Inf. Syst.}, 42(1):181--213, 2015.

\bibitem{YBLG17}
Hao Yin, Austin~R. Benson, Jure Leskovec, and David~F. Gleich.
\newblock Local higher-order graph clustering.
\newblock In {\em Proceedings of the 23rd {ACM} {SIGKDD} International Conference on Knowledge Discovery and Data Mining (KDD)}, pages 555--564, 2017.

\end{thebibliography}
